%% file: main-arxiv.tex
\documentclass[letterpaper,USenglish]{lipics-v2016}

\newif\iflong
\newif\ifshort

\iflong
\else
\shorttrue
\fi
\shortfalse
\longtrue

\usepackage{tikz}
\usetikzlibrary{shapes,arrows,shadows,intersections,calc,through,backgrounds,positioning,matrix,patterns,decorations.markings}
\tikzstyle arrowstyle=[scale=1]
\tikzstyle{separation}=[fill=black, circle,inner sep = 1pt]
\tikzstyle{vertex}=[fill=black, circle,inner sep = 1.5pt]
\tikzstyle{edge} = [draw,-,black]
\tikzstyle anchorEdge=[gray,ultra thick, postaction={decorate,decoration={markings,
    mark=at position 0.75 with {\arrow[arrowstyle, scale=1.5]{stealth}}}}]
\tikzstyle directed=[postaction={decorate,decoration={markings,
    mark=at position 0.75 with {\arrow[arrowstyle, scale=1.3]{stealth}}}}]
\tikzstyle directed2=[postaction={decorate,decoration={markings,
    mark=at position 0.5 with {\arrow[arrowstyle, scale=1.3]{stealth}}}}]
\tikzstyle{canonical} = [edge, gray, directed]
\tikzstyle{fan} = [edge, directed2]
\tikzstyle{shortcut} = [edge, thick, gray, directed]
\tikzstyle{shortcut2} = [edge, thick, gray]

\usepackage{microtype,verbatim}

\newtheorem{proposition}[theorem]{\bf Proposition}
\usepackage{todonotes}

\title{Degree Four Plane Spanners: Simpler and Better}

\author[1]{Iyad Kanj}
\author[1]{Ljubomir Perkovi\'{c}}
\author[1]{Duru T\"{u}rko\u{g}lu}

\affil[1]{School of Computing, DePaul University\\ Chicago, IL, USA\\
\texttt{\{ikanj,lperkovic,dturkogl\}@cs.depaul.edu}}

\authorrunning{I. Kanj, L. Perkovi\'{c}, D. T\"{u}rko\u{g}lu}

\Copyright{Iyad Kanj, Ljubomir Perkovi\'{c}, and Duru T\"{u}rko\u{g}lu}


\newcommand{\tridown}{\bigtriangledown}

\newcommand{\dist}{d}
\newcommand{\dst}{\dist_\tridown}

\newcommand{\dstm}{\delta^\mathrm{min}_\tridown}
\newcommand{\dstw}{\delta^{white}_\tridown}
\newcommand{\dstc}[1]{\delta^{c_#1}_\tridown}
\newcommand{\dstr}{\delta^{red}_\tridown}

\newcommand{\dstb}{\delta^{blue}_\tridown}
\newcommand{\dspan}{\dist_\spanner}

\newcommand{\sector}{\sigma}
\newcommand{\comp}{\mathcal{C}}
\newcommand{\dela}{\mathcal{D}}
\newcommand{\anchiv}{\mathcal{A}}
\newcommand{\spanner}{\mathcal{S}}

\def\ie{{\em i.e.}}

\newcommand{\Oh}{\mathcal{O}}
\newcommand{\Points}{\mathcal{P}}

\begin{document}

\maketitle

\input{abstract}

\input{intro}

\input{defns}
\input{monotone}

\input{algo}

\input{proofs}

\iflong
\input{convex}
\fi
\bibliography{bib}
\end{document}

%% file: abstract.tex
\begin{abstract}
Let $\Points$ be a set of $n$ points embedded in the plane, and let
$\comp$ be the complete Euclidean graph whose point set is $\Points$.
Each edge in $\comp$ between two points $p, q$ is realized as
the line segment $[pq]$, and is assigned a weight equal to the
Euclidean distance $|pq|$. In this paper, we show how to construct in
$\Oh(n\lg{n})$ time a plane spanner of $\comp$ of maximum degree at
most 4 and stretch factor at most 20.  This improves a long sequence
of results on the construction of plane spanners of $\comp$. Our
result matches the smallest known upper bound of 4 by Bonichon et al.
on the maximum degree of plane spanners of $\comp$, while
significantly improving their stretch factor upper bound from 156.82 to 20.
The construction of our spanner is based on Delaunay triangulations
defined with respect to the equilateral-triangle distance, and uses a
different approach than that used by Bonichon et al. Our approach
leads to a simple and intuitive construction of a well-structured
spanner, and reveals useful structural properties of the Delaunay
triangulations defined with respect to the equilateral-triangle
distance.
The structure of the constructed spanner implies that when $\Points$
is in convex position, the maximum degree of this spanner is at most
3. Combining the above degree upper bound with the fact that 3 is a
lower bound on the maximum degree of any plane spanner of $\comp$ when
the point-set $\Points$ is in convex position, the results in this
paper give a tight bound of 3 on the maximum degree of plane spanners
of $\comp$ for point-sets in convex position.
\end{abstract}

%% file: intro.tex
\section{Introduction}
\label{sec:intro}
Let $\Points$ be a set of $n$ points embedded in the plane, and let
$\comp$ be the complete Euclidean graph whose point-set is $\Points$.
Each edge in $\comp$ between two points $p, q$ is realized as
the line segment $[pq]$, and is assigned a weight equal to the
Euclidean distance $|pq|$.
For any constant $\rho \geq 1$, a subgraph $\spanner$ of $\comp$ is a
\emph{$\rho$-spanner} of $\comp$ if for any two points $p, q \in
\Points$, there is a path between $p$ and $q$ in $\spanner$ whose
weight is at most $\rho \cdot |pq|$.  A subgraph $\spanner$ of $\comp$
is a {\em spanner} of $\comp$ if there exists a constant $\rho \geq 1$
such that $\spanner$ is a $\rho$-spanner of $\comp$; the minimum such
$\rho$ is referred to as the \emph{stretch factor} of $\spanner$.
%
%
In this paper, we consider the problem of constructing a plane spanner
of $\comp$ of small degree and small stretch factor. This problem has
received considerable attention, and there is a long list of results
on the construction of plane spanners of $\comp$ that achieve various
trade-offs between the degree and the stretch factor of the spanner.

The problem of constructing a plane spanner of $\comp$ was considered
as early as the 1980's, if not earlier.  Chew~\cite{Che86} proved that
the $L_1$-Delaunay triangulation of $\Points$, which is the Delaunay
triangulation of $\Points$ defined with respect to the $L_1$-distance,
is a $\sqrt{10}$-spanner of $\comp$.
Chew's result was followed by a series of papers showing that other
Delaunay triangulations are plane spanners as well. In 1987, Dobkin
{\it et al.}~\cite{DFS90} showed that the classical $L_2$-Delaunay
triangulation of $\Points$ is a plane spanner of stretch factor
at most
$\frac{\pi(1+\sqrt{5})}{2}$. This bound was subsequently improved by
Keil and Gutwin~\cite{KG92} to $\frac{4\pi}{3\sqrt{3}}$. In the
meantime, Chew~\cite{Che89} showed that the $TD$-Delaunay
triangulation defined using a distance function based on an
equilateral triangle --- rather than a square ($L_1$-distance) or a circle
($L_2$-distance) --- is a $2$-spanner. This result was generalized by
Bose et al.~\cite{BCCS10}, who showed that the Delaunay triangulation
defined with respect to any convex distance function (\ie, based on a
convex shape) is a spanner of~$\comp$. The bound on the stretch factor
of the $L_2$-Delaunay triangulation by Keil and Gutwin stood
unchallenged for many years until Xia recently improved the bound to
below 2~\cite{Xia13}.  Recently as well, Bonichon {\em et
al.}~\cite{BGHP12} improved Chew's original bound on the stretch
factor of the $L_1$-Delaunay triangulation to $\sqrt{4+2\sqrt{2}}$,
and showed this bound to be tight.

All the Delaunay triangulations mentioned above can have
unbounded degree. In recent years, bounded degree plane spanners have
been used as the building block of wireless network topologies.
Wireless distributed system technologies, such as wireless ad-hoc and
sensor networks, are often modeled as proximity graphs in the
Euclidean plane. Spanners of proximity graphs represent topologies
that can be used for efficient communication. For these applications,
in addition to having low stretch factor, spanners are typically
required to be plane and have bounded degree, where both requirements
are useful for efficient routing~\cite{planerouting,boundedrouting}.

The wireless network applications motivated researchers to turn their
attention to minimizing the maximum degree of the plane spanner as
well as its stretch factor.  It can be readily seen that 3 is a lower
bound on the maximum degree of a spanner of $\comp$, because a
Hamiltonian path through a set of points arranged in a grid has
unbounded stretch factor. Work on bounded degree but not necessarily
plane spanners of $\comp$ closely followed the above-mentioned work on
plane spanners. In a 1992 breakthrough, Salowe \cite{Sal94} proved the
existence of spanners of maximum degree 4. The question was then
resolved by Das and Heffernan~\cite{DH96} who showed that spanners of
maximum degree 3 always exist. The focus in this line of research was
to prove the existence of low degree spanners and the techniques
developed to do so were not tuned towards constructing spanners that
had both low degree and low stretch factor.  For example, the bound on
the stretch factor of the degree 4 spanner in Salowe \cite{Sal94} is
greater than $10^9$, which is far from practical.  Furthermore, the
bounded-degree spanners shown to exist are not guaranteed to be
plane.

Bose {\it et al.}~\cite{BGS05a} were the first to show how to extract
a subgraph of the $L_2$-Delaunay triangulation that is a
bounded-degree, plane spanner of $\comp$. The maximum degree and
stretch factor they obtained were subsequently improved by Li and
Wang~\cite{LW04}, by Bose {\it et al.}~\cite{BSX09}, and by Kanj and
Perkovi\'{c}~\cite{KP08} (see Table~\ref{ta:related2}). The approach
used in all these results was to extract a bounded degree spanning
subgraph of the $L_2$-Delaunay triangulation and the main goal was to
obtain a bounded-degree plane spanner of $\comp$ with the smallest
possible stretch factor. In a breakthrough result,
Bonichon {\it et al.}~\cite{BGHP10} lowered
the bound on the maximum degree of a plane spanner from 14 to 6.
Instead of using the $L_2$-Delaunay triangulation as the starting
point of the spanner construction, they used the Delaunay
triangulation based on the equilateral-triangle distance, defined
originally by Chew~\cite{Che89}, and exploited a connection between
these Delaunay triangulations and $\frac{1}{2}$-$\theta$ graphs.
The plane spanner they constructed also
has a small stretch factor of 6.  Independently, Bose {\it et
al.}~\cite{BCC12} were also able to obtain a plane spanner of maximum
degree at most 6, by starting from the $L_2$-Delaunay triangulation.
Recently, Bonichon et al.~\cite{BKPX15} were able to construct a plane
spanner of degree at most 4 and stretch factor at most 156.82. Their
construction is based on the Delaunay triangulation defined with
respect to the $L_1$ norm. Most of the above spanner constructions can
be performed in time $\Oh(n\lg{n})$, where $n$ is the number of points
in $\Points$.

\begin{table}
\begin{center}
\begin{tabular}{lrr}
{\bf Paper} & {\bf $\Delta$} & {\bf Stretch factor bound} \\ \hline
Bose {\it et al.}~\cite{BGS05a} & 27 & $(\pi+1) C_0 \approx 8.27$ \\
\hline
Li and Wang~\cite{LW04} & 23 & $(1+\pi \sin \frac{\pi}{4}) C_0 \approx
6.43$ \\ \hline
Bose {\it et al.}~\cite{BSX09} & 17 & $(2+2\sqrt{3} + \frac{3\pi}{2} +
2\pi\sin(\frac{\pi}{12})) C_0 \approx 23.56 $\\ \hline
Kanj and Perkovi\'{c}~\cite{KP08} & 14 &
$(1+\frac{2\pi}{14\cos(\frac{\pi}{14})}) C_0 \approx 2.91$\\ \hline
Bonichon {\it et al.}~\cite{BGHP10} & 6 & 6 \\ \hline
Bose {\it et al.}~\cite{BCC12} & 6 &
$1/(1-\tan(\pi/7)(1+1/\cos(\pi/14)))C_0 \approx 81.66$ \\ \hline
Bonichon {\it et al.}~\cite{BKPX15} & 4 & $156.82$ \\ \hline
{\bf This paper} & {\bf 4} & {\bf 20} \\ \hline
\end{tabular}
\end{center}
\caption{Results on plane spanners with maximum degree bounded by
$\Delta$. The constant $C_0 = 1.998$ is the best known upper bound on
the stretch factor of the $L_2$-Delaunay triangulation~\cite{Xia13}.}
\label{ta:related2}
\end{table}

In this paper, we present a construction of a plane spanner
$\spanner$ of $\comp$ of degree at most 4 and stretch factor at
most 20.
This result matches the smallest known upper bound of 4 on the maximum
degree of the spanner by Bonichon et al.~\cite{BKPX15}, while
significantly improving their stretch factor bound from 156.82 to 20.
Our construction is also simpler and more intuitive. It is based on Delaunay
triangulations defined with respect to the equilateral-triangle
distance, just like the degree 6 spanner construction used by Bonichon et
al.~\cite{BGHP10}, which could be viewed as the starting point of our
construction. To get down to maximum degree 4, our approach introduces fresh
techniques in both the construction and the analysis of the spanner.
Unlike the approach in~\cite{BGHP10}, our approach has a bias towards
certain edges of the Delaunay triangulation; this bias ensures that the constructed spanner is well structured. To make up for edges not in the spanner, we make use of
recursion which, unlike the construction in~\cite{BKPX15}, may have depth not bounded by a
constant. To ensure that the recursion is controlled and yields short paths,
we aggressively add shortcut edges to the spanner to ensure the existence
of paths with specific monotonicity properties, which we refer to as monotone weak paths. Finally, in our analysis we
use a new type of distance metric and we also take the extra step of analyzing
the stretch factor of our spanner with respect to $\comp$ directly, rather than
with respect to the underlying Delaunay triangulation.

 \ifshort The structure of the spanner, attained due to biasing certain Delaunay edges, implies that if the given point-set is in convex position, then the constructed spanner has maximum degree at most 3. Therefore, for any point-set in convex position, there exists a plane spanner of $\comp$ of maximum degree at most 3.  We can also show that 3 is a lower bound on the maximum degree of plane spanners of $\comp$ for point-sets in convex position. The preceding implies that 3 is a tight bound on the maximum degree of plane spanners of $\comp$ for point-sets in convex position, and this completely and satisfactorily resolves the question about the maximum degree of plane geometric spanners of $\comp$ for point-sets in convex position. Due to the lack of space, the formal statement and proof of the aforementioned result, as well as some other proofs in the paper, are omitted and can be found in a full version of the paper available at: **** \fi

\iflong The structure of the spanner, attained due to biasing certain Delaunay edges, allows us to prove that if the given point-set is in convex position, then the constructed spanner has maximum degree at most 3. Therefore, for any point-set in convex position, there exists a plane spanner of $\comp$ of maximum degree at most 3.  We also show that 3 is a lower bound on the maximum degree of plane spanners of $\comp$ for point-sets in convex position. The preceding implies that 3 is a tight bound on the maximum degree of plane spanners of $\comp$ for point-sets in convex position, and this completely and satisfactorily resolves the question about the maximum degree of plane geometric spanners of $\comp$ for point-sets in convex position. \fi

%% file: defns.tex
\section{Preliminaries}
\label{sec:prelim}

Given a set of points $\Points$ embedded in the Euclidean plane, we
consider the complete weighted graph $\comp(\Points)$, or simply
$\comp$, where each edge between any two points $p,q \in \Points$ is
associated with the line segment $[pq]$, and is assigned a weight
equal to the Euclidean distance $|pq|$.
Given a subgraph $G$ of $\comp$, we define $G$ to be a {\em plane}
subgraph if the edges of $G$ do not cross each other, \ie, the line
segments associated with the edges of $G$ intersect only at their
endpoints. The {\em maximum degree} of $G$ is
the maximum degree (in $G$) over all points in $\Points$; we say that
a family of graphs has {\em bounded degree} if there is an integer constant $c \geq 0$ such that every
graph in the family has a maximum degree at most $c$.
If graph $G$ is connected, we define the {\em distance} between any two
points $p,q \in \Points$, denoted $\dist_G(p,q)$, to be the weight of
a minimum-weight path between $p$ and $q$ in $G$, where the weight of
a path is the sum of the weights of its edges.

Given a constant $\rho \geq 1$, we say that $G$ is a {\em
$\rho$-spanner} of $\comp$ if for any two points $p,q \in \Points$,
$\dist_G(p, q) \leq \rho \cdot |pq|$. We also say that a family of
geometric graphs, one for every finite set $\Points$ of points in the
plane, is a spanner if there is a constant $\rho \geq 1$ such that
every graph $G(\Points)$ in the family
%
%
is a $\rho$-spanner of $\comp(\Points)$; we refer to the minimum such
constant $\rho$ as the {\em stretch factor} of the family.  In this
paper, the family we consider consists of the spanners obtained by
applying our algorithm on all finite point-sets in the plane.
%

In this paper, we rely on a metric that is different from the
Euclidean metric.  In order to define this metric, we fix an
equilateral triangle with two of its points lying on the $x$-axis at
coordinates $(0,0)$ and $(1,0)$, and the third point lying below the $x$-axis; we use the symbol $\tridown$ to refer to this equilateral
triangle.  We define a triangle to be a {\em $\tridown$-homothet} if
it can be obtained through a translation of $\tridown$ followed
by a scaling. We define the {\em triangular metric}, $\dst$, as follows:

\begin{definition}
\label{def:metric}
For any two points $p, q \in \Points$, we define $\dst(p,q)$ to be the
side-length of the smallest $\tridown$-homothet that contains $p$ and
$q$ on its boundary; we denote this triangle $\tridown(p,q)$.
\end{definition}

It is easy to verify that $\dst$ is indeed a metric. In particular,
for any two points $p,q$, we have $\dst(p,q)=0 \Leftrightarrow p=q$,
we have symmetry as in $\dst(p,q) = \dst(q,p)$, and for any third point
$r$, we have the triangle inequality $\dst(p,q) \leq \dst(p,r) + \dst(r,q)$.
It is also easy to see that $p$ or $q$ must be a vertex of the
triangle $\tridown(p,q)$ and that $|pq| \leq \dst(p,q)$.

Using the triangular metric $\dst$, we define a subgraph $\dela$ of $\comp$
as follows. For every point $w \in \Points$, we partition the space around $w$ into six
equiangular cones whose common apex is $w$, three above and three below the
horizontal line passing through $w$, as illustrated in
Figure~\ref{fi:definitions}-(a). We denote the middle cone
above the horizontal line and the two outer cones below the horizontal
line as the {\em positive} cones of $w$, and the remaining
three cones as the {\em negative} cones of $w$.
Each point $w$ chooses an edge in each of its three positive cones by
selecting the point $v \neq w$ in the cone such that $\dst(w, v)$ is
minimum. Assuming that $\Points$ is in general position\footnote{$\Points$ is in general position if no pair of points $v,w
\in \Points$ lie on a line parallel to any of the boundary lines
defining the six cones. We note that it is always possible to rotate the
equilateral triangle that defines the metric $\dst$ to ensure that
the finite set $\Points$ is in general position and so the results in this
paper hold for all sets of points and not just for points in general position.}, for any
two distinct points $v, v'$ in a positive cone of $w$, we obtain $\dst(w,
v) \neq \dst(w, v')$. Let $\dela$ be the graph whose vertex-set is
$\Points$ and whose edge-set is the set of edges selected as described.

We make the following observation regarding the graph $\dela$. The {\em
$\frac{1}{2}$-$\theta$ graph} of $\Points$ is the graph whose
point-set is $\Points$, and whose edges are obtained as follows: at
each point $w$, and for each of the three positive cones of apex $w$,
select the edge $wv$ in the cone where $v$ is the point whose
projection distance to the angular bisector of the cone is minimum.
Bonichon et al.~\cite{BGHI10} showed that the $\frac{1}{2}$-$\theta$ graph
of $\Points$ is the same as the {\em TD-Delaunay triangulation} of
$\Points$~\cite{Che89} defined based on the empty triangle property:
there is an edge between two point $v, w \in P$
if there exists a homothet of $\tridown$ containing $v$ and $w$ on its
boundary whose interior is empty of points of $\Points$. It is easy to see
that the $\frac{1}{2}$-$\theta$ graph of $\Points$ coincides with the
graph $\dela$ defined above, and hence with the TD-Delaunay
triangulation\footnote{A TD-Delaunay triangulation of
$\Points$ is not necessarily a triangulation of $\Points$ as defined
traditionally (a triangulation of the convex hull of the set of points).
Just as Chew~\cite{Che89} did, we abuse the term {\em triangulation}
because TD-Delaunay triangulations are closely related to classical
$L_2$-Delaunay triangulations.}.

For convenience, we label the positive cones at each point of $\Points$,
in clockwise order and starting with the positive cone above the horizontal
line, {\em red}, {\em green}, and {\em blue}; we
also label the negative cones, in clockwise order and starting with the
negative cone below the horizontal line, red, green, and blue.
We assign an orientation and a color to the edges of
$\dela$ by orienting each edge outwards from the point $w$ that selects it and
by coloring it red, blue, or green depending on whether the edge lies in
the positive red, blue, or green cone of point $w$, as illustrated in
Figure~\ref{fi:definitions}-(a). We emphasize that the edge
orientations are only used for the purpose of constructing the spanner and
proving its desired properties; the final spanner in our construction is
an undirected graph obtained by removing edge orientations. In fact, we
will abuse terminology and, throughout the paper, use the term {\em path}
to refer to weak paths in $\dela$; we will always use the term
{\em directed path} when edge orientations are relevant.

We observe that for any point $w \in \Points$ there is at most one edge outgoing from $w$
in a positive cone of $w$, but there can be an unbounded number
of edges incoming to $w$ in a negative cone of $w$, and that in such
cases all these
edges have the same color as the cone itself (e.g.,
see Figure~\ref{fi:definitions}-(b)). We follow the same approach as
Bonichon et al.~\cite{BGHP10}, and identify in each negative cone of point $w$
an edge that plays a key role in the spanner construction:

\begin{definition} \rm
\label{def:anchor}

For any point $w \in \Points$, and for each negative cone of $w$ that contains
at least one
edge incoming to $w$, let (directed) edge $(v, w) \in \dela$ be the edge in the
cone such that $\dst(v, w)$ is minimum. We define $(v, w)$ to be the
{\em anchor} of $w$ in the cone.

\end{definition}
We  say that anchors incident to the same point
$w$ are {\em adjacent} if their cones are adjacent. Note that for any two adjacent anchors incident to $w$, one of the
two adjacent anchors must lie in a positive cone of $w$ and must be an anchor
of a point other than $w$.

Consider a negative cone of a point $w \in \Points$ containing at
least one incoming edge to $w$ in $\dela$. Let $(v_1, w), \ldots,
(v_k, w) \in \dela$, where $k \geq 1$, be all the incoming edges to
$w$ that lie in the cone, listed in counterclockwise order, as
illustrated in Figure~\ref{fi:definitions}-(b), and let $(v_j, w)$,
for some $j$ such that $1 \leq j \leq k$, be the anchor of $w$ in the
cone. We call $(v_1, w), \ldots, (v_k, w)$ the {\em fan} of the anchor
$(v_j, w)$. We identify $(v_j,w)$ as the anchor of each edge in the
fan. Note that every edge in $\dela$ has an anchor which could be
itself. We call the first edge $(v_1, w)$ and the last edge $(v_k, w)$
of the fan the {\em boundary edges} of the anchor $(v_j, w)$.  Note
that either one (possibly both) of the boundary edges of an anchor
could be the anchor itself.
If $k \geq 2$, since $\dela$ is a triangulation, it follows that
$(v_i, v_{i+1}, w)$ is a triangle in $\dela$, for $i=1, \ldots, k-1$.
Hence, $v_1, \ldots, v_k$ is a (weak) path in $\dela$ between the
endpoints $v_1$ and $v_k$.  We call this path the {\em canonical path}
of $w$ in the designated cone; we also call each edge on this path a {\em
canonical edge} of $w$.
Finally, we refer to the (weak) subpath $v_r, \ldots, v_s$ of the
canonical path $v_1, \ldots, v_k$ of $w$ as the canonical
path between $v_r$ and $v_s$ of $w$.
The two {\em sides} of an edge are the two
half-planes defined by the line obtained by extending the edge. We say
that a canonical edge $e$ is canonical {\em on a side} of $e$ if $e$ is a
canonical edge of a point that lies on that side of $e$. Note that a canonical edge can be canonical on both sides.

%
%
%

\renewcommand{\theenumi}{\alph{enumi}}
We state the following easy to verify facts without proof:
\begin{lemma}
\label{lem:canonical-properties}
Let $(s,t)$ be a canonical edge of a point $w$, and let $(s',t)$ be the anchor of
$(s,t)$.
\begin{enumerate}

\item The edges $(s,w)$ and $(t,w)$ are in $\dela$.

\item The edge $(s,w)$ cannot be a canonical edge on the side containing $t$.

\item The edge $(t,w)$ is not an anchor.

\item The edge $(s,t)$ is a boundary edge of its anchor $(s',t)$.

\end{enumerate}
\end{lemma}
\renewcommand{\theenumi}{\arabic{enumi}}

\input{figures/defns-fig.tex}

%% file: figures/defns-fig.tex
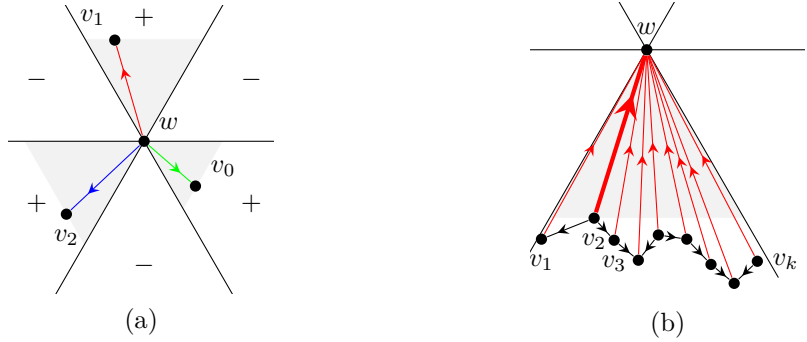
\begin{figure}
\noindent
\begin{minipage}{.5\textwidth}
\centering
\begin{tikzpicture}[scale = 0.75]

\clip(-2.4,-2.7) rectangle (2.3,2.4);


\filldraw[fill=black!5, draw=white] (0,0) -- ++(60:2.1) -- ++ (180:2.1) -- cycle; 
\filldraw[fill=black!5, draw=white] (0,0) -- ++(0:1.4) -- ++ (-120:1.4) -- cycle; 
\filldraw[fill=black!5, draw=white] (0,0) -- ++(-120:2.13) -- ++ (120:2.13) -- cycle;

\draw (-5,0) -- (5,0);
\draw (0,0) -- (60:5); \draw (0,0) -- (240:5);
\draw (0,0) -- (120:5); \draw (0,0) -- (300:5);

\draw (0,0) ++(30:2.2) node {$-$};
\draw (0,0) ++(90:2.2) node {$+$};
\draw (0,0) ++(150:2.2) node {$-$};
\draw (0,0) ++(210:2.2) node {$+$};
\draw (0,0) ++(270:2.2) node {$-$};
\draw (0,0) ++(330:2.2) node {$+$};

\draw (0,0) node[vertex,label=30:{$w$}] (w) {};
\draw (0.9105,-0.79449) node[vertex,label=4:{$v_0$}] (v0) {};
\draw (-0.52001,1.79464) node[vertex,label=95:{$v_1$}] (v1) {};
\draw (-1.37780,-1.29337) node[vertex,label=-90:{$v_2$}] (v2) {};

\draw [canonical,green] (w) -- (v0);
\draw [canonical,red] (w) -- (v1);
\draw [canonical,blue] (w) -- (v2);
\end{tikzpicture}

\centering (a)
\end{minipage}%
\begin{minipage}{.5\textwidth}
\centering

\begin{tikzpicture}[scale = 0.7]

\clip(-2.2,-4.7) rectangle (3,0.9);

\filldraw[fill=black!5, draw=white] (0,0) -- ++(-120:3.7) -- ++ (0:3.7) -- cycle;

\draw (-5,0) -- (5,0);
\draw (0,0) -- (60:5); \draw (0,0) -- (240:5);
\draw (0,0) -- (120:5); \draw (0,0) -- (300:5);

\draw (0,0) node[vertex,label=90:{$w$}] (w) {};

\draw (-2,-3.6) node[vertex,label=-90:{$v_1$}] (v1) {};
\draw (-1,-3.2) node[vertex,label=-90:{$v_2$}] (v) {};
\draw (-0.62,-3.62) node[vertex,label=-90:{$v_3$}] (v3) {};
\draw (-0.16,-4) node[vertex,label=-90:{}] (v4) {};
\draw (0.22,-3.52) node[vertex,label=90:{}] (v5) {};
\draw (0.76,-3.6) node[vertex,label=90:{}] (v6) {};
\draw (1.22,-4.08) node[vertex,label=-135:{}] (v7) {};
\draw (1.66,-4.44) node[vertex,label=-90:{}] (v8) {};
\draw (2.1,-4.02) node[vertex,label=0:{$v_k$}] (v9) {};




 


\draw [anchorEdge,red] (v) -- (w);
\draw [canonical,black] (v) -- (v1);
\draw [canonical,black] (v) -- (v3);
\draw [canonical,black] (v3) -- (v4);
\draw [canonical,black] (v5) -- (v4);
\draw [canonical,black] (v5) -- (v6);
\draw [canonical,black] (v6) -- (v7);
\draw [canonical,black] (v7) -- (v8);
\draw [canonical,black] (v9) -- (v8);

\draw [fan,red] (v1) -- (w);
\draw [fan,red] (v3) -- (w);
\draw [fan,red] (v4) -- (w);
\draw [fan,red] (v5) -- (w);
\draw [fan,red] (v6) -- (w);
\draw [fan,red] (v7) -- (w);
\draw [fan,red] (v8) -- (w);
\draw [fan,red] (v9) -- (w);
\end{tikzpicture}

\centering (b)
\end{minipage}%

\caption{(a) To construct graph $\dela$, every point $w$ chooses the shortest
edge, according to the $\dst$ distance, in every positive cone. (b) Edge
$(v_2,w)$ is the anchor of $w$ in the negative cone shown because it is the
shortest edge, according to the $\dst$ distance, among edges incoming to
$w$ in the cone; the path $v_1, \ldots, v_k$ is the canonical path of anchor
$(v_2,w)$.}
\label{fi:definitions}

\end{figure}

%% file: monotone.tex
\section{Monotone (weak) paths}
\label{sec:monotone}
We define next a type of path in $\dela$ that generalizes canonical paths
and that will be a key tool in our construction. We give two equivalent definitions of such a path; proving the equivalence between these two definitions is easy.

\begin{definition}
Let $v$ be a point lying in the positive cone of $u$ whose color is
$c$. A (weak) path in $\dela$ between $u$ and $v$ is {\em monotone} if
the path is bi-colored, with $c$ being one of the colors, and the path
satisfies the two equivalent properties:
\begin{itemize}
\item After reversing the direction of all edges not colored $c$, the path
is directed from $u$ to $v$.
\item No two consecutive edges of the path lie in neighboring cones of
the shared endpoint.
\end{itemize}
\label{def:monotone}
\end{definition}


\input{figures/defns2-fig.tex}

The key property of a monotone path between $u$ and $v$ is that its length
can be bounded by twice the side-length of $\tridown(u,v)$, \ie, by
$2\dst(u,v)$. This follows from a stronger insight which we develop
next. To facilitate our discussion, we label the vertices of a
$\tridown$-homothet {\em green}, {\em blue}, and
{\em red}, in clockwise order starting from the upper left
vertex.
Let $v$ be a point lying in a positive cone
of $u$ of color $c_1$. With $u$ being the vertex of $\tridown(u,v)$ of color $c_1$, let the remaining two vertices of $\tridown(u,v)$ be
$y$ and $z$, and let $c_2$ be the color of $y$ and $c_3$ be the color of
$z$.

\begin{definition}
We define (refer to Figure~\ref{fig:defns2} where $c_1=green$,
$c_2 =red$, and $c_3=blue$):
\label{def:blue-white-dist}
\begin{enumerate}
\item $\dstc{2}(u,v) = \dstc{2}(v,u) = |yv|$.
\item $\dstc{3}(u,v) = \dstc{3}(v,u) = |zv|$.
\item $\dstm(u,v) = \dstm(v,u) = \min\{\dstc{2}(u,v),
\dstc{3}(u,v)\}$.
\end{enumerate}

\end{definition}


Let $P$ be a monotone path in $\dela$ between $u$ and $v$, and let the edges of
$P$ be colored $c_1$ or $c_2$. We define, using the lines $zv$ and $zu$ as axes
of a coordinate system of the Euclidean plane, {\em the projection onto $zv$
and $zu$} (see Figure~\ref{fig:defns2}-(b)). In the following lemma, we use this
projection to map the edges of $P$ to upper bound the length of $P$:

\renewcommand{\theenumi}{\alph{enumi}}
\begin{lemma}
Let $v$ be a point lying in the positive cone of $u$ of color $c_1$. For any monotone path $P$ between $u$ and $v$ whose edges are colored
$c_1$ or $c_2$ (refer to Figure~\ref{fig:defns2} where $c_1=green$,
$c_2 =red$, and $c_3=blue$):
\begin{enumerate}
\item The projections of all edges of $P$ onto $zu$ (resp., $zv$) do not overlap and are contained within the
segment $[zu]$ (resp., $[zv]$); see
Figures~\ref{fig:defns2}-(b) and~\ref{fig:defns2}-(c).

\item If $(p,q)$ is an edge of $P$ colored $c_1$ (resp., $c_2$) then the projection
onto $zu$ (resp., $zv$) of $(p,q)$ has length $\dst(p,q) \geq |pq|$.

\item The sum of the lengths of the edges of $P$ colored $c_1$ is
at most $\dst(u,v)=|zu|$.

\item The sum of the lengths of the edges of $P$ colored $c_2$ is
at most $\dstc{3}(u,v)=|zv|$.

\item The length of $P$ is at most $2\dst(u,v)$.
\end{enumerate}

\label{lem:monotone}
\end{lemma}

\begin{proof}
For part {\em (a)}, we consider the coordinates
of the points of $P$ in the coordinate system of the Euclidean plane defined
by using the lines $zv$ and $zu$ as axes. When visiting the points of $P$
in the order in which they appear on $P$, the
coordinates of the points along the $zu$ (resp., $zv$) axis form a monotonic
sequence (decreasing or increasing) between the coordinates of $u$ and $z$
(resp. $z$ and $v$), and part {\em (a)} follows. Since $zu$ is parallel to an edge of $\tridown(p, q)$,
and hence the projection of $(p, q)$ onto $zu$ has length $\dst(p, q)$, part {\em (b)} follows.
Parts {\em (c)} and {\em (d)} follow from parts {\em (a)} and {\em (b)}, and  part {\em (e)} follows
from parts {\em (c)} and {\em (d)}.
\end{proof}

The following lemma, illustrated in Figure~\ref{fig:defns3}-(a), is
an insight implicit in Lemma 2 of~\cite{BGHP10} that is
implied by Lemma~\ref{lem:monotone}.

\begin{lemma}
For any two edges $(v,w)$ and $(u,w)$ that lie in the same fan:
\begin{enumerate}

\item The canonical path $P$ between $v$ and $u$ is monotone.

\item The sum of the lengths of all monochromatic edges on $P$
 is at most $\dst(v,u)$.

\item The length of the canonical path $P$ between $v$ to $u$ is at most
$2\dst(v,u)$.
\end{enumerate}
\label{lem:canonical-bound}
\end{lemma}

\begin{proof}
For part {\em (a)}, we assume, without loss of generality, that $w$
lies in the red positive cones of $v$ and of $u$. We then observe that
for every point $p$ on $P$, $(p,w)$ is an edge in $\dela$. Therefore,
every edge of $P$ must lie in the blue or green positive cones of its
tail, and thus path $P$ is bi-colored.  Furthermore, since $\dela$ is
planar, at every intermediate point $p$ of $P$, the two edges of $P$
incident to $p$ must lie in non-adjacent cones.  The canonical path
$P$ between $v$ and $u$ is thus monotone.  Hence, parts (b) and (c)
follow by Lemma~\ref{lem:monotone}.
\end{proof}

\input{figures/defns3-fig.tex}
\renewcommand{\theenumi}{\arabic{enumi}}

%% file: figures/defns2-fig.tex
\begin{figure}[!b]
\begin{minipage}[b]{.25\textwidth}
\centering

\begin{tikzpicture}[scale = 0.6]
\clip(-2.2,-7.6) rectangle (6,-2.25);

\draw (-1,-3.2) node[vertex,label=92:{$u$}] (q) {};
\draw (2,-4.52) node[vertex,label=0:{$v$}] (p) {};

\path [name path=qy] (q) -- ++(300:5);
\path [name path=qx] (q) -- ++(0:5);

\path [help lines,name path=pu] (p) -- ++(60:2.4);
\path [help lines,name path=pl] (p) -- ++(240:5);

\path [name intersections={of=qx and pu, by={[label=90:$z$]Y}}];
\path [name intersections={of=qy and pl, by={[label=180:$y$]Z}}];



\draw [blue, very thick,densely dotted] (p)  -- (Y);
\draw [red, very thick,densely dotted] (p)  -- (Z);

\draw [blue,very thick] (q) -- (Y);
\draw [red, very thick] (q) -- (Z);
\draw (3.7,-4.02) node[blue] {$\dstb(u,v)$};
\draw (2.8,-5.7) node[red] {$\dstr(u,v)$};
\draw (1.2,-2.72) node[blue] {$\dst(u,v)$};
\draw (-1.1,-5.2) node[red] {$\dst(u,v)$};



\end{tikzpicture}

\centering (a)

\end{minipage}
\begin{minipage}[b]{.34\textwidth}
\centering

\begin{tikzpicture}[scale = 0.7]

\clip(-2.4,-7.6) rectangle (4.6,-2.5);


\draw (-1,-3.2) node[vertex,label=135:{$u$}] (v) {};
\path [help lines,name path=vx] (v) -- ++(0:5.4);
\path [help lines,name path=vd] (v) -- ++(300:5.4);

\draw (0.25,-4) node[vertex,label=180:{$p_2$}] (p2) {};
\draw (-0.1,-5.5) node[vertex,label=180:{$p_1$}] (p1) {};

\draw (2.1,-6.02) node[vertex,label=-45:{$v$}] (p0) {};

\path [help lines,name path=p0y] (p0) -- ++(60:5);
\path [help lines,name path=p0y2] (p0) -- ++(240:5);





\path [draw,help lines,name path=p2y] (p2) -- ++(60:1.1);
\path [draw,help lines,name path=p2x] (p2) -- ++(0:3.2);

\path [draw,help lines,name path=p1y] (p1) -- ++(60:2.8);
\path [draw,help lines,name path=p1x] (p1) -- ++(0:2.65);




\path [name intersections={of=vd and p0y2, by={[label=180:$y$]V}}];
\path [name intersections={of=vx and p0y, by={[label=88:$z$]W}}];

\draw [gray!25] (v) -- (V) -- (W) -- (v);

\path [name intersections={of=vx and p2y, by={q2}}];
\path [name intersections={of=p2x and p1y, by={q1}}];
\path [name intersections={of=p1x and p0y, by={q0}}];

\draw [blue,densely dotted,thick] (q2) -- (p2) -- (q1);
\draw [blue,thick] (q1) -- (p1) -- (q0);
\path [name intersections={of=vx and p2y, by={P2Z}}];
\path [name intersections={of=vx and p1y, by={P1Z}}];
\path [name intersections={of=p0y and p2x, by={P2Y}}];
\path [name intersections={of=p0y and p1x, by={P1Y}}];

\draw [blue,thick] (P2Z) -- (v);
\draw [blue,densely dotted,thick] (P1Y) -- (p0);




\draw (0,-2.9) node[blue] {{\scriptsize $\dst(u,p_2)$}};
\node at (2.3,-2.8) [inner sep = 1pt,blue] (A1) {{\scriptsize $\dstb(u,p_2)$}};
\node at (2.1,-3.65) [inner sep = 1pt,blue] (A2) {{\scriptsize $\dstb(p_1,p_2)$}};
\draw (1.4,-4.55) node[blue] {{\scriptsize $\dst(p_1,p_2)$}};
\draw (1.5,-5.22) node[blue] {{\scriptsize $\dst(p_1,v)$}};
\draw (3.4,-5.75) node[blue] {{\scriptsize $\dstb(p_1,v)$}};

\draw [blue,->] (A1) .. controls (1, -2.8) and (1, -3.6) .. (0.55,-3.6);
\draw [blue,->] (A2) .. controls (0.8, -3.7) .. (0.55,-3.9);


\draw [canonical,green] (v) -- (p2);
\draw [canonical,red] (p1) -- (p2);
\draw [canonical,green] (p1) -- (p0);
\end{tikzpicture}

\centering (b)
\end{minipage}
\begin{minipage}[b]{.40\textwidth}
\centering

\begin{tikzpicture}[scale = 0.7]

\clip(-2.1,-7.6) rectangle (5.2,-2.4);


\draw (-1,-3.2) node[vertex,label=90:{$u$}] (v) {};
\path [help lines,name path=vx] (v) -- ++(0:5.4);
\path [help lines,name path=vd] (v) -- ++(300:5.4);

\draw (0.25,-4) node[vertex,label=180:{$p_2$}] (p2) {};
\draw (-0.1,-5.5) node[vertex,label=180:{$p_1$}] (p1) {};

\draw (2.1,-6.02) node[vertex,label=-45:{$v$}] (p0) {};

\path [help lines,name path=p0y] (p0) -- ++(60:5);
\path [help lines,name path=p0y2] (p0) -- ++(240:5);





\path [draw,help lines,name path=p2y] (p2) -- ++(60:1.1);
\path [draw,help lines,name path=p2x] (p2) -- ++(0:3.2);

\path [draw,help lines,name path=p1y] (p1) -- ++(60:2.8);
\path [draw,help lines,name path=p1x] (p1) -- ++(0:2.65);




\path [name intersections={of=vd and p0y2, by={[label=180:$y$]V}}];
\path [name intersections={of=vx and p0y, by={[label=88:$z$]W}}];

\draw [gray!25] (v) -- (V) -- (W) -- (v);

\path [name intersections={of=vx and p2y, by={q2}}];
\path [name intersections={of=p2x and p1y, by={q1}}];
\path [name intersections={of=p1x and p0y, by={q0}}];

\draw [blue,densely dotted,thick] (q2) -- (p2) -- (q1);
\draw [blue,thick] (q1) -- (p1) -- (q0);
\path [name intersections={of=vx and p2y, by={P2Z}}];
\path [name intersections={of=vx and p1y, by={P1Z}}];
\path [name intersections={of=p0y and p2x, by={P2Y}}];
\path [name intersections={of=p0y and p1x, by={P1Y}}];

\draw [blue,very thick] (P2Z) -- (v);
\draw [blue,densely dotted,very thick] (P2Z) -- (P1Z);
\draw [blue,very thick] (P1Z) -- (W);
\draw [blue,densely dotted,very thick] (W) -- (P2Y);
\draw [blue,very thick] (P2Y) -- (P1Y);
\draw [blue,densely dotted,very thick] (P1Y) -- (p0);







\draw (4,-4.65) node[blue] {$\dstb(u,v)$};
\draw (1.3,-2.72) node[blue] {$\dst(u,v)$};

\draw [canonical,green] (v) -- (p2);
\draw [canonical,red] (p1) -- (p2);
\draw [canonical,green] (p1) -- (p0);
\end{tikzpicture}

\centering (c)
\end{minipage}

\caption{(a) If point $v$ lies in the positive cone of point $u$, and the vertices of $\tridown(u, v)$, $u, y, z$, are colored green, red, and
blue, respectively, then $\dst(u,v)=|zu|$ and $\dstb(u,v) =|zv|$. 
(b) and (c) $P$ is a monotone path between $u$ and $v$ with edges colored green or red.
The projection onto $zu$ (resp.~$zy$) of $(u,p_2)$,
$(p_1,p_2)$, and $(p_1,v)$ do not overlap and are
contained within $[zu]$ (resp.~$[zv]$).}
\label{fig:defns2}

\end{figure}
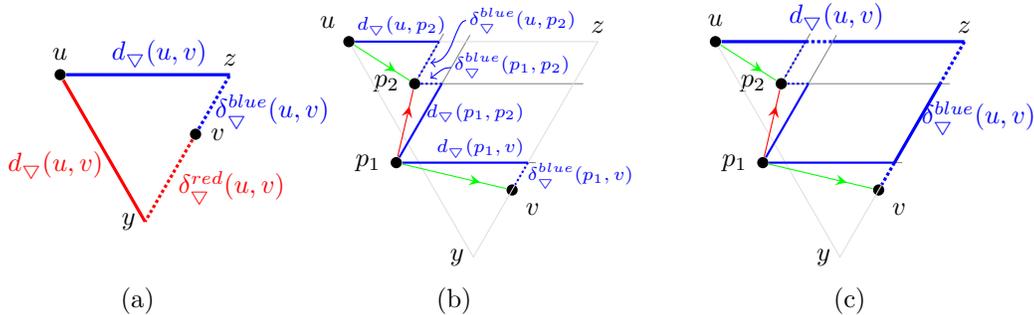

%% file: figures/defns3-fig.tex
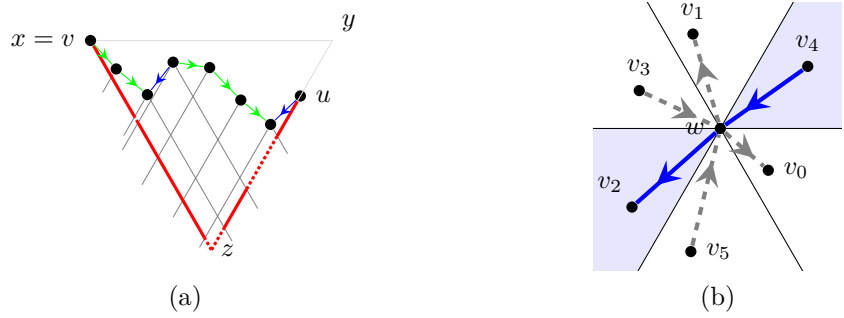
\begin{figure}
\begin{minipage}[b]{.5\textwidth}
\centering

\begin{tikzpicture}[scale = 0.9]

\clip(-2.2,-6.6) rectangle (3,-2.4);


\draw (0,0) node[vertex,label=90:{$w$}] (w) {};

\draw (-1,-3.2) node[vertex,label=180:{$x=v$}] (v) {};
\draw (-0.62,-3.62) node[vertex,label=180:{}] (p6) {};
\draw (-0.16,-4) node[vertex,label=90:{}] (p5) {};
\draw (0.22,-3.52) node[vertex,label=90:{}] (p4) {};
\draw (0.76,-3.6) node[vertex,label=45:{}] (p3) {};
\draw (1.22,-4.08) node[vertex,label=45:{}] (p2) {};
\draw (1.66,-4.44) node[vertex,label=0:{}] (p1) {};
\draw (2.1,-4.02) node[vertex,label=0:{$u$}] (p0) {};

\path [help lines,name path=vy] (v) -- ++(300:5);
\path [help lines,name path=vx] (v) -- ++(0:5);

\path [draw,help lines,name path=p6z] (p6) -- ++(240:0.4);
\path [draw,help lines,name path=p5z] (p5) -- ++(240:0.7);
\path [draw,help lines,name path=p4z] (p4) -- ++(240:1.5);
\path [draw,help lines,name path=p3z] (p3) -- ++(240:2);
\path [draw,help lines,name path=p2z] (p2) -- ++(240:2);
\path [draw,help lines,name path=p1z] (p1) -- ++(240:2.1);
\path [draw,help lines,name path=p1y] (p1) -- ++(300:0.5);
\path [draw,help lines,name path=p4y] (p4) -- ++(300:2.5);
\path [draw,help lines,name path=p5y] (p5) -- ++(300:2.5);

\path [help lines,name path=p0u] (p0) -- ++(60:2.4);
\path [help lines,name path=p0l] (p0) -- ++(240:5);

\path [name intersections={of=vx and p0u, by={[label=45:$y$]Y}}];
\path [name intersections={of=vy and p0l, by={[label=0:$z$]Z}}];

\draw [gray!25] (v) -- (Y) -- (Z) -- (v);

\path [draw,  name intersections={of=vy and p6z, by={P6Z}}];
\path [name intersections={of=vy and p5z, by={P5Z}}];
\path [name intersections={of=vy and p4z, by={P4Z}}];
\path [name intersections={of=vy and p3z, by={P3Z}}];
\path [name intersections={of=vy and p2z, by={P2Z}}];
\path [name intersections={of=vy and p1z, by={P1Z}}];

\draw [red, very thick] (v) -- (P5Z);
\draw [red, densely dotted, very thick] (P4Z) -- (P5Z);
\draw [red, very thick] (P4Z) -- (P1Z);
\draw [red, densely dotted, very thick] (Z) -- (P1Z);

\path [name intersections={of=p0l and p1y, by={P1Y}}];
\path [name intersections={of=p0l and p4y, by={P4Y}}];
\path [name intersections={of=p0l and p5y, by={P5Y}}];

\draw [red, densely dotted, very thick] (Z) -- (P5Y);
\draw [red, very thick] (p0)  -- (P1Y);
\draw [red, densely dotted, very thick] (P4Y) -- (P1Y);
\draw [red, very thick] (P4Y) -- (P5Y);


\draw [canonical,green] (v) -- (p6);
\draw [canonical,green] (p6) -- (p5);
\draw [canonical,color=blue] (p4) -- (p5);
\draw [canonical,green] (p4) -- (p3);
\draw [canonical,green] (p3) -- (p2);
\draw [canonical,green] (p2) -- (p1);
\draw [canonical,color=blue] (p0) -- (p1);
\end{tikzpicture}

\centering (a)
\end{minipage}
\begin{minipage}[b]{.5\textwidth}
\centering

\begin{tikzpicture}[scale = 0.7]

\clip(-2.4,-2.7) rectangle (2.3,2.4);

\filldraw[fill=blue!10, draw=white] (0,0) -- ++(5cm,0mm) arc [start angle=0, end angle=60, radius=5cm] -- cycle;
\filldraw[fill=blue!10, draw=white] (0,0) -- ++(-5cm,0mm) arc [start angle=180, end angle=240, radius=5cm] -- cycle;

\draw (-5,0) -- (5,0);
\draw (0,0) -- (60:5); \draw (0,0) -- (240:5);
\draw (0,0) -- (120:5); \draw (0,0) -- (300:5);

\draw (0,0) node[vertex,label=180:{$w$}] (w) {};
\draw (0.9105,-0.79449) node[vertex,label=0:{$v_0$}] (v0) {};
\draw (-0.52001,1.79464) node[vertex,label=90:{$v_1$}] (v1) {};
\draw (-1.67780,-1.49337) node[vertex,label=93:{$v_2$}] (v2) {};
\draw (-1.54,0.72) node[vertex,label=90:{$v_3$}] (v3) {};
\draw (1.66,1.18) node[vertex,label=90:{$v_4$}] (v4) {};
\draw (-0.56,-2.34) node[vertex,label=0:{$v_5$}] (v5) {};

\draw [anchorEdge,dashed] (w) -- (v0);
\draw [anchorEdge,dashed] (w) -- (v1);
\draw [anchorEdge,blue] (w) -- (v2);
\draw [anchorEdge,dashed] (v3) -- (w);
\draw [anchorEdge,blue] (v4) -- (w);
\draw [anchorEdge,dashed] (v5) -- (w);
\end{tikzpicture}

\centering (b)
\end{minipage}

\caption{(a) Illustration of Lemma~\ref{lem:canonical-bound}.
(b) $w$ is incident to an anchor in every cone. In step 1, both of the
blue anchors are added.  In step 2, on the other hand, no more than
one white anchor is added below the horizontal line through $w$ and no
more than one white anchor is added above. }
\label{fig:defns3}

\end{figure}

%% file: algo.tex
\section{The Spanner}
\label{sec:spanner}
In this section, we describe the construction of a plane spanner
of $\comp$ of maximum degree at most 4 and stretch factor at most 20.
In our construction, we will bias blue -- positive
and negative -- cones and edges. This bias results in a spanner satisfying structural properties that allow us to prove the desired upper bounds on the spanner degree
and stretch factor. This bias also ensures that the spanner has maximum degree at most 3 when the point-set $\Points$ is in convex position.
We find it convenient to refer to the four non-blue cones, as well as all the red and green edges, as
{\em white}, as illustrated in Figure~\ref{fig:defns3}-(b), and
introduce some terminologies.

If $e$ is a canonical edge of a point $w$ that lies in a white (resp., blue)
cone of $w$, we say that $e$ is a canonical edge {\em in a white}
(resp., {\em blue}) {\em cone}. We note that a canonical edge
could be in a white cone of one point and in a blue
cone of another. Given a white anchor $(v,w)$, the ray starting
from $w$ extending $(v, w)$ partitions the (white) negative cone of $w$
containing $v$ into two {\em sides}: we refer to the side of the cone
that is adjacent to a blue cone as the {\em blue side}, and we refer
to the other side that is adjacent to a white cone as the {\em white
side}. We say that an edge $(u, w)$ in the fan of $(v, w)$ is on the {\em white side} (resp.~{\em blue side}) if it is on the white side (resp.~blue side) of $(v,w)$.


The following describes the construction of the spanner $\spanner$
of $\comp$. The construction is based on the underlying triangulation
$\dela$ of $\comp$. We start by constructing a degree-4 anchor subgraph $\anchiv$ of $\spanner$ that includes all blue anchors. We then finalize $\spanner$
by adding some white canonical edges and shortcuts to ensure the reconstructibility of all canonical paths.    \\

{\bf Constructing the anchor subgraph $\anchiv$ of $\spanner$:}

\begin{enumerate}

\item We add to $\anchiv$ every blue anchor (of $\dela$).

\item In increasing order of length with respect to the metric $\dst$,
for every white anchor $a$, we add $a$ to~$\anchiv$ if no white anchor adjacent
to $a$ is already in $\anchiv$.

\end{enumerate}

{\bf Reconstructing canonical paths in blue
cones (see Figure~\ref{fi:algorithm}-(a)):}

\begin{enumerate}

\setcounter{enumi}{2}

\item We add to $\spanner$ every (white) canonical edge in a blue cone
if the edge is not in $\anchiv$.

\item For every pair of canonical edges $(p,q), (r,q)$ in a blue cone
such that $(p,q), (r,q) \in \spanner \setminus \anchiv$, we add to
$\spanner$ the shortcut edge $(p,r)$, color it white, and remove
$(p,q)$ and $(r,q)$ from $\spanner$.

\end{enumerate}

{\bf Reconstructing canonical paths in white cones (see
Figure~\ref{fi:algorithm}-(b)):}

\begin{enumerate}

\setcounter{enumi}{4}

\item We add to $\spanner$ every white canonical edge that is on the white
side of its (white) anchor, but only if its anchor is not in $\anchiv$.

\item For every white anchor $(v, w)$ and its boundary edge
$(u, w) \neq (v, w)$ on the white side, let $P$ be the canonical path
($u=p_0, p_1, \ldots, p_k=v$). We apply the following procedure at
a current point $p_i$ starting with $i=0$ and stopping when $i=k$:

\begin{enumerate}

  \item If the canonical edge $(p_{i+1}, p_i)$ is white, we skip this
edge and set $i$ to $i+1$;

  \item Otherwise, $(p_i, p_{i+1})$ must be blue. Let $j > i$ be the largest
	index of a point on $P$ such that the line segment $[p_ip_j]$ does not
	intersect the canonical path from $p_i$ to $p_j$ (except at $p_i$ and
	$p_j$). We add the shortcut $(p_j, p_i)$ to $\spanner$ and color it
	white; we remove the (white) canonical edge $(p_j, p_{j-1})$
	from $\spanner$ if $(p_j, p_{j-1}) \in \spanner \setminus \anchiv$;
	and we set $i$ to $j$.
\end{enumerate}
\end{enumerate}

\input{figures/algo-fig.tex}

In the following section we prove that this algorithm yields a plane
spanner of maximum degree at most $4$ and stretch factor at most $20$. We
provide here a high-level overview of our arguments.

To show planarity, we note that the underlying graph $\dela$ is planar and that
the only edges of $\spanner$ not in $\dela$ are the shortcut edges added in
steps 4 and 6.b. We will prove in Lemmas~\ref{lem:uncrossed}
and~\ref{lem:plane} that each such edge does not intersect any other edge of
$\spanner$. For the degree upper bound, we note that the first two steps of the algorithm
yield the subgraph $\anchiv$ of maximum degree at most 4. In the remaining steps,
we carefully add additional edges, whether canonical edges or shortcuts of
canonical paths. To prove the degree bound, we develop a charging argument
that assigns each edge of $\spanner$ to a cone at each endpoint and show,
in Lemma~\ref{lem:degree}, that no more than 4 cones are charged at every point.

Note that step 1 adds all blue anchors to $\spanner$.
For white anchors, we cannot include all of them in $\spanner$ because of the degree bound. However, by considering anchors in an increasing order of their length in step 2, we
can ensure that a white anchor is not added to $\spanner$ only if an adjacent shorter white anchor was added to $\spanner$.
This ordering is crucial because it allows us to use short anchors in $\spanner$ to
reconstruct longer anchors not in
$\spanner$.

In steps 3 and 4, we add to $\spanner$ all (white)
canonical edges in blue cones, except for some consecutive pairs of
canonical edges that are replaced with shortcut edges. Therefore, steps 1 through 4 of the algorithm ensure that every canonical path in a blue cone has been reconstructed.  In particular, since all blue anchors are in $\spanner$, every blue boundary edge, and therefore every blue canonical edge, has been reconstructed.

Consider now a canonical path $P$ in a white negative cone of $w$, and the two subpaths $P_w$ and $P_b$ of $P$ on the white side and the blue side of this cone, respectively.
In step 6, we take shortcuts to obtain, together with step 5, a white monotone path to ensure that there is a short path reconstructing $P_w$.
As for $P_b$, we argued above that the blue edges of $P_b$ are reconstructible. As for the white edges of $P_b$, we note that they are on the white side of their anchors, and hence,
the argument used to reconstruct $P_w$ above applies to them as well. The above, once again with step 5, ensures the reconstructibility of $P_b$, and hence of $P$.

%% file: figures/algo-fig.tex
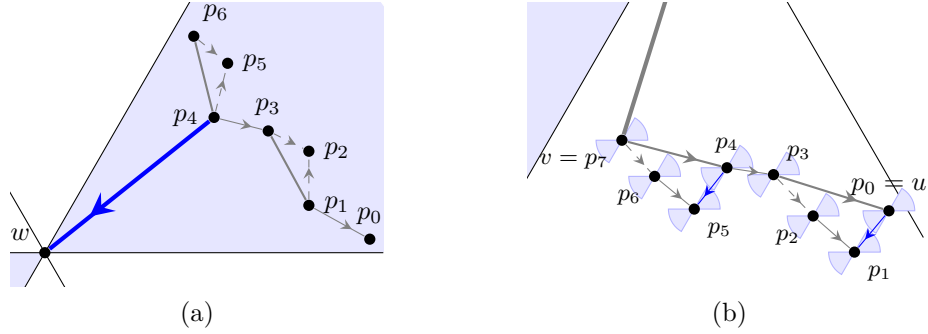
\begin{figure}
\noindent
\begin{minipage}[b]{.5\textwidth}
\centering

\begin{tikzpicture}[scale = 0.9]

\clip(-.5,-.5) rectangle (5,3.7);

\filldraw[fill=blue!10, draw=white] (0,0) -- ++(6.2cm,0mm) arc [start angle=0, end angle=60, radius=6.2cm] -- cycle;
\filldraw[fill=blue!10, draw=white] (0,0) -- ++(-6.2cm,0mm) arc [start angle=180, end angle=240, radius=6.2cm] -- cycle;

\draw (-5,0) -- (5,0);
\draw (0,0) -- (60:5); \draw (0,0) -- (240:5);
\draw (0,0) -- (120:5); \draw (0,0) -- (300:5);

\draw (0,0) node[vertex,label=150:{$w$}] (w) {};
\draw (4.8,0.2) node[vertex,label=90:{$p_0$}] (v0) {};
\draw (3.9,0.7) node[vertex,label=0:{$p_1$}] (v1) {};
\draw (3.9, 1.5) node[vertex,label=0:{$p_2$}] (v2) {};
\draw (3.3,1.8) node[vertex,label=90:{$p_3$}] (v3) {};
\draw (2.5,2) node[vertex,label=180:{$p_4$}] (v4) {};
\draw (2.7,2.8) node[vertex,label=0:{$p_5$}] (v5) {};
\draw (2.2,3.2) node[vertex,label=87:{$p_6$}] (v6) {};


\draw [canonical,dashed] (v4) -- (v5);
\draw [canonical,dashed] (v6) -- (v5);
\draw [canonical] (v4) -- (v3);
\draw [canonical,dashed] (v3) -- (v2);
\draw [canonical,dashed] (v1) -- (v2);
\draw [canonical] (v1) -- (v0);

\draw [shortcut2] (v3) -- (v1);
\draw [shortcut2] (v4) -- (v6);

\draw [anchorEdge,blue] (v4) -- (w);
\end{tikzpicture}

\centering (a)
\end{minipage}
\begin{minipage}[b]{.5\textwidth}
\centering

\begin{tikzpicture}[scale = 1.145]

\clip(-2.1,-4.9) rectangle (2.6,-1.6);

\filldraw[fill=blue!10, draw=white] (0,0) -- ++(5cm,0mm) arc [start angle=0, end angle=60, radius=5cm] -- cycle;
\filldraw[fill=blue!10, draw=white] (0,0) -- ++(-5cm,0mm) arc [start angle=180, end angle=240, radius=5cm] -- cycle;

\draw (-5,0) -- (5,0);
\draw (0,0) -- (60:5); \draw (0,0) -- (240:5);
\draw (0,0) -- (120:5); \draw (0,0) -- (300:5);

\draw (0,0) node[vertex,label=90:{$w$}] (w) {};

\filldraw[fill=blue!10, draw=blue!30] (-1,-3.2) -- ++(3mm,0mm) arc [start angle=0, end angle=60, radius=3mm] -- cycle;
\filldraw[fill=blue!10, draw=blue!30] (-1,-3.2) -- ++(-3mm,0mm) arc [start angle=180, end angle=240, radius=3mm] -- cycle;
\draw (-1,-3.2) node[vertex,label=225:{\small $v=p_7$}] (v) {};

\filldraw[fill=blue!10, draw=blue!30] (-0.62,-3.62) -- ++(3mm,0mm) arc [start angle=0, end angle=60, radius=3mm] -- cycle;
\filldraw[fill=blue!10, draw=blue!30] (-0.62,-3.62) -- ++(-3mm,0mm) arc [start angle=180, end angle=240, radius=3mm] -- cycle;
\draw (-0.62,-3.62) node[vertex,label=225:{\small $p_6$}] (p6) {};

\filldraw[fill=blue!10, draw=blue!30] (-0.16,-4) -- ++(3mm,0mm) arc [start angle=0, end angle=60, radius=3mm] -- cycle;
\filldraw[fill=blue!10, draw=blue!30] (-0.16,-4) -- ++(-3mm,0mm) arc [start angle=180, end angle=240, radius=3mm] -- cycle;
\draw (-0.16,-4) node[vertex,label=315:{\small $p_5$}] (p5) {};

\filldraw[fill=blue!10, draw=blue!30] (0.22,-3.52) -- ++(3mm,0mm) arc [start angle=0, end angle=60, radius=3mm] -- cycle;
\filldraw[fill=blue!10, draw=blue!30] (0.22,-3.52) -- ++(-3mm,0mm) arc [start angle=180, end angle=240, radius=3mm] -- cycle;
\draw (0.22,-3.52) node[vertex,label=90:{\small $p_4$}] (p4) {};

\filldraw[fill=blue!10, draw=blue!30] (0.76,-3.6) -- ++(3mm,0mm) arc [start angle=0, end angle=60, radius=3mm] -- cycle;
\filldraw[fill=blue!10, draw=blue!30] (0.76,-3.6) -- ++(-3mm,0mm) arc [start angle=180, end angle=240, radius=3mm] -- cycle;
\draw (0.76,-3.6) node[vertex,label=45:{\small $p_3$}] (p3) {};

\filldraw[fill=blue!10, draw=blue!30] (1.22,-4.08) -- ++(3mm,0mm) arc [start angle=0, end angle=60, radius=3mm] -- cycle;
\filldraw[fill=blue!10, draw=blue!30] (1.22,-4.08) -- ++(-3mm,0mm) arc [start angle=180, end angle=240, radius=3mm] -- cycle;
\draw (1.22,-4.08) node[vertex,label=225:{\small $p_2$}] (p2) {};

\filldraw[fill=blue!10, draw=blue!30] (1.7,-4.5) -- ++(3mm,0mm) arc [start angle=0, end angle=60, radius=3mm] -- cycle;
\filldraw[fill=blue!10, draw=blue!30] (1.7,-4.5) -- ++(-3mm,0mm) arc [start angle=180, end angle=240, radius=3mm] -- cycle;
\draw (1.7,-4.5) node[vertex,label=-45:{\small $p_1$}] (p1) {};

\filldraw[fill=blue!10, draw=blue!30] (2.1,-4.02) -- ++(3mm,0mm) arc [start angle=0, end angle=60, radius=3mm] -- cycle;
\filldraw[fill=blue!10, draw=blue!30] (2.1,-4.02) -- ++(-3mm,0mm) arc [start angle=180, end angle=240, radius=3mm] -- cycle;
\draw (2.1,-4.02) node[vertex,label=90:{$p_0 = u$}] (p0) {};

\draw [shortcut] (p3)-- (p0);
\draw [shortcut] (v)-- (p4);

\draw [anchorEdge] (v) -- (w);
\draw [canonical,dashed] (v) -- (p6);
\draw [canonical] (p6) -- (p5);
\draw [canonical,color=blue] (p4) -- (p5);
\draw [canonical] (p4) -- (p3);
\draw [canonical,dashed] (p3) -- (p2);
\draw [canonical,-] (p2) -- (p1);
\draw [canonical,color=blue] (p0) -- (p1);

\end{tikzpicture}

\centering (b)
\end{minipage}

\caption{(a) In
step 3, white canonical edges of $w$ in the negative blue cone of $w$ are
added to
$\spanner$  if not in $\anchiv$ already; in step 4, any pair of canonical
edges of $w$  added in step 3 that are incoming at the same point are
replaced by a shortcut between the outgoing endpoints ($(p_6,p_5)$ and 
$(p_4,p_5)$ replaced by shortcut $(p_6,p_4)$ and $(p_3,_2)$ and 
$(p_1,p_2)$ replaced by shortcut $(p_3,p_1)$.) (b) Shortcut edges $(p_3,p_0)$
and $(p_7,p_4)$ are added to $\spanner$
in step 6; edges $(p_7,p_6)$ and $(p_3,p_2)$
are not in $\spanner$ unless they are anchors in $\anchiv$.}
\label{fi:algorithm}

\end{figure}

%% file: proofs.tex
\section{Properties of the Spanner}
\label{sec:spanner-props}
In this section, we prove the three properties of the spanner
obtained using our algorithm: planarity, the maximum degree upper
bound of 4, and the stretch factor bound of 20. We start with the
following justification for coloring shortcut edges added in step 6
white:

\begin{lemma}
For every shortcut edge $(p_j,p_i)$ added to $\spanner$ in step 6,
$p_j$ and $v=p_k$ both lie in the same negative white cone of $p_i$,
and they both lie in the same negative white cone of $p_{j-1}$.
\end{lemma}

\begin{proof}
Because $\dst(w,p_k) < \dst(w,p_i)$ and $(p_i,w) \in \dela$, $p_k$
must lie in a negative white cone of $p_i$. The lemma thus holds if $j=k$.
Otherwise, by the choice of $p_j$, $p_j$ must lie on the same side of line $p_ip_k$ as point $w$;
again, because $(p_i,w) \in \dela$, $p_j$ must lie in a (negative) white cone
of $p_i$ that also contains $p_k$. Similar arguments apply to $p_{j-1}$.
\end{proof}

Next, we show that $\spanner$ is plane. We first need the following definition and lemma.

\begin{definition}
\label{def:uncrossed}
An edge $(u,w) \in \dela$ is {\em uncrossed} if no shortcut in
$\spanner$ crosses $(u,w)$.
\end{definition}

\begin{lemma}
\label{lem:uncrossed}
All anchors, all canonical edges, and all boundary edges are
uncrossed.
\end{lemma}

\begin{proof}
Let $(p,r)$ be a shortcut that was added in step 4 of the spanner
construction, let $(p,q)$ and $(r,q)$ be the pair of canonical edges
in the blue cone as described in step 4, and let $w$ be the apex of
this blue cone. It is easy to verify that $(q,w) \in
\dela$ is the only edge in $\dela$ that $(p,r)$ crosses, and that
$(q,w)$ is not a boundary edge, a canonical edge, or an anchor.
Next, consider a shortcut $(p_j,p_i)$ that was added in
step 6 of the spanner construction, and let $(v,w)$ be the white
anchor and $p_{i+1}, p_{i+2}, \ldots, p_{j-1}$ be the points on the
canonical path between $p_i$ and $p_j$ as described in step 6.
Again, it is easy to verify that $(p_{i+1},w), (p_{i+2},w), \ldots,
(p_{j-1},w)$ are the only edges in $\dela$ that the shortcut
$(p_j,p_i)$ crosses, and that none of them is a boundary edge, a
canonical edge, or an anchor.
\end{proof}

\begin{lemma}
The subgraph $\spanner$ is a plane subgraph of $\comp$.
\label{lem:plane}
\end{lemma}

\begin{proof}
Let $\spanner_1=\dela \cap \spanner$ and $\spanner_2 = \spanner
\setminus \spanner_1$. Note that $\spanner_1$ consists of $\anchiv$ plus
the canonical edges added in steps (3) or (5) that are kept after steps
(4) and (6), and that $\spanner_2$ consists only of the shortcuts added in
steps (4) and (6).  Since $\spanner_1$ is a subgraph of $\dela$,
$\spanner_1$ is plane.	By Lemma~\ref{lem:uncrossed}, all the edges in $\spanner_1$ are uncrossed, \ie, no shortcut (edge
in $\spanner_2$) crosses an edge in $\spanner_1$.  To conclude the
proof, we show that no two edges in $\spanner_2$ cross either.
Observe that any two shortcuts connect pairs of endpoints of canonical
paths that either belong to different fans, and in which case they cannot cross,
or belong to the same fan. In the latter case, the shortcuts do
not cross because shortcuts always connect the endpoints of
non-overlapping canonical paths.
\end{proof}

To facilitate the discussion in the proof of the degree upper bound,
we refer to the two adjacent white cones above (resp., below) the
horizontal line through a point $p \in \Points$ as the {\em upper}
(resp., {\em lower}) {\em white sector} of $p$; we also refer to the two blue
cones at $p$ as the {\em left} and {\em right blue sectors} of $p$. We develop next
a charging scheme and use it to show that, for each point $p$, each
edge incident to $p$ in $\spanner$ can be mapped in a one-to-one
fashion to one of the four sectors at~$p$. To describe the charging
scheme for every edge $e \in \spanner$ and for every endpoint $p$ of
$e$, we define $\sector(e,p)$ to be the sector of $p$ that contains
$e$. Also for a point $p$, we denote by $LB_p$, $RB_p$, $UW_p$, and
$LW_p$, the left blue, the right blue, the upper white, and the
lower white sectors of $p$ respectively. We describe in the table
below the charging scheme for every edge $e=(x,y) \in \spanner$
based on which step of the construction $e$ is added to $\spanner$.\\
\begin{tabular}{c c c c}
\\ {\bf Step} & {\bf Classification of $e = (x,y)$} &
{\bf Charge at $x$} & {\bf Charge at $y$} \\
1 & Blue anchor in $\anchiv$ &
$\sector(e,x)=LB_x$ & $\sector(e,y)=RB_y$ \\
2 & White anchor in $\anchiv$ &
$\sector(e,x)=UW_x$ or $LW_x$ & $\sector(e,y)=LW_y$ or $UW_y$ \\
3 & White canonical edge in a blue cone &
$\sector(e,x)=UW_x$ or $LW_x$ & $LB_y$ \\
4 & Shortcut in a blue cone &
$\sector(e,x)=UW_x$ or $LW_x$ & $\sector(e,y)=LW_y$ or $UW_y$ \\
5 & White canonical edge in a white cone &
$RB_x$ & $\sector(e,y)=UW_y$ or $LW_y$ \\
6 & Shortcut in a white cone &
$RB_x$ & $\sector(e,y)=UW_y$ or $LW_y$ \\ &
\end{tabular}

\begin{lemma}
Let $p \in \Points$. The charging scheme above charges each edge incident to $p$ in $\spanner$ to one of the sectors at $p$ such that each sector is charged with at most one edge.
\label{lem:charging}
\end{lemma}

\renewcommand{\theenumi}{\roman{enumi}}

\begin{proof}
Let $p \in \Points$. We will show that the charging scheme above charges every edge incident to $p$ in $\spanner$ to one of the four sectors at $p$ in a one-to-one fashion.

First, consider the left blue sector of $p$. Observe that $LB_p$ could potentially be charged in each of the
following two situations:

\begin{enumerate}
\item In step (1) by a blue anchor $(x=p,y) \in \anchiv$; and
\item in step (3) by a white canonical edge $(x,y=p)$ in a blue cone.
\end{enumerate}

First, observe that (i) and (ii) cannot apply simultaneously by
part (c) of Lemma~\ref{lem:canonical-properties}.  Second, observe
that $LB_p$ cannot be charged twice according to situation (ii)
because this means that there are two white canonical edges incoming
to $p$ in a blue cone, which would imply that step (4) of the spanner
construction applies, and both incoming edges to $p$ would be removed
from $\spanner$. It follows that $LB_p$ is charged at most once.

Next, we consider the right blue sector of $p$. Observe that $RB_p$ could potentially be charged in each of the
following three situations:

\begin{enumerate}

\item In step (1) by a blue anchor $(x,y=p) \in \anchiv$;
\item in step (5) by a white canonical edge $(x=p,y)$ in a white cone; and
\item in step (6) by a shortcut $(x=p,y)$ in a white cone.
\end{enumerate}

First, observe that in both situations (ii) and (iii), we know that
there is a white canonical edge $(p,q)$ in a white cone with apex $w$,
that is $(p,w)$ and $(q,w)$ are white edges in $\dela$; in (ii)
$(p,q)$ is $(x,y)$; and in (iii) $(p,q)$ is the (last) edge of the
canonical path between $x=p$ and $y$, in step (6) of the spanner
construction, by the choice of point $x$ in this step.

The existence of this white canonical edge implies that $RB_p$ must be
empty of edges of $\dela$, that (i) cannot apply when either (ii) or
(iii) applies.

Second, by part (b) of Lemma~\ref{lem:canonical-properties}, $(p,w)$
is not a canonical edge on the side that contains $q$, \ie, $(p,w)$
cannot be a canonical edge in a white cone.

Therefore, $(p,q)$ is the only white canonical edge in a white cone
that is outgoing from $p$, and hence, $RB_p$ cannot be charged twice
according to (ii), or twice according to (iii).

Finally, (ii) and (iii) cannot apply simultaneously because step
(6) of the construction removes this white canonical edge $(p,q)$ when
it adds the shortcut $(x,y)$.

It follows that $RB_p$ is charged at most once.

Finally, we consider the upper sector of $p$, and note
that the arguments for the lower sector of $p$ follow similarly. Observe that $UW_p$ could potentially be charged in each of the
following five situations:

\begin{enumerate}

\item In step (2) by a white anchor $(x,y) \in \anchiv$, $x=p$ or $y=p$;

\item in step (3) by a white canonical edge $(x=p,y)$ in a blue cone;

\item in step (4) by a shortcut $(x,y)$ in a blue cone, $x=p$ or
$y=p$;

\item in step (5) by a white canonical edge $(x,y=p)$ in a white cone;
and

\item in step (6) by a shortcut $(x,y=p)$ in a white cone.

\end{enumerate}

First, since step (2) of the construction disallows adjacent anchors
to be in $\anchiv$, $UW_p$ cannot be charged twice according to (i).

Second, we show that in both situations (ii) and (iii), none of the
other situations (i), (iv), or (v) apply.  To show this, we observe
that there exists a white canonical edge $(p,q) \notin \anchiv$ in a
blue cone with apex $w$, that is $(p,w)$ and $(q,w)$ are blue edges in
$\dela$; in (ii) $(p,q)$ is $(x,y)$ and in (iii) $(p,q)$ is the
canonical edge incident to $p$ that step (4) of the construction
removes after adding the shortcut $(x,y)$.

The existence of this white canonical edge $(p,q)$ implies that the
negative white cone of $UW_p$ must be empty of edges of $\dela$, hence,
(iv) cannot apply.

By part (b) of Lemma~\ref{lem:canonical-properties}, since $(p,w)$
is not a canonical edge on the side that contains $q$, (v) cannot
apply.

Also, since $(p,q) \notin \anchiv$ is the only outgoing edge in $UW_p$,
(i) cannot apply either.

Next, observe that $UW_p$ cannot be charged twice according to (ii)
or according to (iii) simply because $(p,q)$ is the only canonical
edge in a blue cone that lies in $UW_p$.

Finally, observe that (ii) and (iii) cannot apply simultaneously
because as described above step (4) of the construction removes
this canonical edge $(p,q)$.

Assuming that (ii) and (iii) do not apply, we analyze
situation (iv).  In this case, we know that the white canonical edge
$(x,p)$ is in a white cone with apex $w$, that is $(x,w)$ and $(p,w)$
are white edges in $\dela$.

Since step (5) of the construction ensures that the anchor of the
canonical edge $(x,p)$ is not in $\anchiv$, and also since by part (c)
of Lemma~\ref{lem:canonical-properties}, $(p,w)$ is not anchor, thus
not in $\anchiv$, we conclude that (i) does not apply.

Next, since $(x,p)$ is a white and not a blue canonical edge for the
fan that contains $(x,w)$ and $(p,w)$, step (6) of the
construction ensures that (v) does not apply.

Finally, observe that $UW_p$ cannot be charged twice according to
(iv) simply because $(x,p)$ is the only canonical edge in a white cone
that lies in $UW_p$.

For the remaining cases, we analyze situation (v).	In this case,
there is a blue canonical edge $(p,q)$ in a white cone
with apex $w$, that is $(p,w)$ and $(q,w)$ are white edges in $\dela$.
The existence of this blue canonical edge implies that the
white cone contained in $UW_p$ is empty of edges of $\dela$, and
because step (6) adds a shortcut incoming at $p$, we know that the
only outgoing white edge $(p,w) \in UW_p$ is not an anchor, thus (i)
cannot apply.  Finally, we cover all the cases by observing that $UW_p$
cannot be charged twice according to (v) simply because step (6)
of the construction adds at most one shortcut for any blue canonical
edge, and because $(p,q)$ is the only blue (canonical) edge outgoing
from $p$.

Therefore, each of of $UW_p$ and $LW_p$ is charged at most once.

It follows from the above that each of the four sectors at
$p$ is charged with at most one edge incident to $p$ in $\spanner$. This
completes the proof.
\end{proof}
\renewcommand{\theenumi}{\arabic{enumi}}

\begin{lemma}
The maximum degree of $\spanner$ is at most 4.
\label{lem:degree}
\end{lemma}

\begin{proof}
The statement of the lemma directly follows from Lemma~\ref{lem:charging} because there are four sectors at each point $p \in \Points$.
\end{proof}

The remainder of this section is devoted to showing that the stretch
factor of $\spanner$ is at most 20.  We do so by first proving a
sequence of lemmas that derive upper bounds on the distance in $\spanner$
between the endpoints of different types of edges in $\dela$; we then use
these lemmas to derive a bound on the stretch factor of $\spanner$.

\begin{lemma}
For any uncrossed blue edge $(u,w) \in \dela$, $\dspan(u,w) \leq 3
\dst(u,w)$.
\label{lem:blue-edge-uncrossed}
\end{lemma}
\begin{proof}
Let $(v,w)$ be the blue anchor of the blue edge $(u,w)$.	In
step (1) of the algorithm, we add all the blue anchors in $\spanner$,
and thus $(v,w) \in \spanner$.  Also, in step (3) of the construction, we
add in $\spanner$ all the canonical edges in blue cones except that, in
step (4), we substitute some pairs of these canonical edges with shortcuts.
Since $(u,w)$ is uncrossed, these canonical edges and shortcuts
provide a path for connecting $v$ and $u$.  Using the triangle
inequality, this path that includes the shortcuts is not longer than
the canonical path between $v$ and $u$.  Hence, in the worst case, we
may assume, without loss of generality, that the path connecting $v$ and
$u$ consists only of canonical edges on the canonical path.  This
canonical path plus the anchor constitutes a path between $u$~and~$w$.
By Lemma~\ref{lem:canonical-bound}, the length of this canonical path
is bounded by $2\dst(v,u) \leq 2\dst(u,w)$.  We also have that $|vw|
\leq \dst(u,w)$.  Consequently, the length of this path is bounded by
$\dst(u,w)$ (anchor) plus $2 \dst(u,w)$ (canonical
path).  It follows that $\dspan(u,w) \leq 3 \dst(u,w)$.
\end{proof}

\begin{lemma}
For any white anchor $(v,w)$ and any uncrossed white edge $(u,w) \in
\dela$ that lies on the white side of $(v,w)$, $\dspan(v,u) \leq
\dst(v,u) + \dstb(v,u) \leq 2\dst(v,u)$.  Furthermore, if $(v,w) \in
\spanner$, then $\dspan(u,w) \leq \dst(u,w) + \dstb(u,w) \leq 2\dst(u,w)$.
\label{lem:white-white}
\end{lemma}
\begin{proof}
We describe below how to construct a white monotone path in $\spanner$
between $u$ and $v$.
If $(v,w) \in \spanner$, we extend this path to a white monotone path
between $u$ and $w$.
Then, we obtain the desired bounds using Lemma~\ref{lem:monotone}.

To describe the white monotone path between $u$ and $v$, we consider
the uncrossed edges of $\dela$ on the fan of $(v,w)$ whose endpoints lie on the
canonical path between $v$ and $u$.  We observe that shortcuts and
white canonical edges connect the (distinct) endpoints of those
uncrossed edges, and that they form a white monotone path between $u$
and $v$ because at each point the white edges of the path incident to
the point lie on opposite sides of the horizontal line through the point.
We call this white monotone path the {\em white monotone
connection} between $v$ and $u$.  We know that all of the shortcuts on
this white monotone connection are in $\spanner$ and even though some
of the white canonical edges may not be in $\spanner$, we know, for
all such white canonical edges, that we have their anchors in
$\spanner$.
For each such white canonical edge $(s,t)$, we recursively expand the
current white monotone path by including the anchor $(r,t)$ of $(s,t)$
and by including the white monotone connection between $r$ and $s$.  We
observe that the path obtained after the expansion of $(s,t)$
continues to be a white monotone path.
%
%
Therefore, by recursively expanding this path for white canonical
edges that are not in $\spanner$, we obtain a white monotone path
between $v$ and $u$.  Furthermore, if $(v,w) \in \spanner$, we expand
this path to include the white anchor $(v,w)$ while preserving its
monotonicity.
\end{proof}

\begin{lemma}
For any white anchor $(v,w)$ and any white edge $(u,w) \in \dela$ that
lies on the blue side of $(v,w)$, $\dspan(v,u) \leq 5\dst(v,u)$.
Furthermore, if $(v,w) \in \spanner$, then $\dspan(u,w) \leq 6\dst(u,w)$.
\label{lem:white-blue}
\end{lemma}
\begin{proof}
The canonical path from $v$ to $u$ consists of blue and white canonical
edges. The total length of the blue canonical edges
does not exceed $\dst(v,u)$, and the total length of the white canonical edges
does not exceed $\dst(v,u)$ by Lemma~\ref{lem:canonical-bound}.
By Lemma~\ref{lem:uncrossed}, we know that all of these canonical
edges are uncrossed.  By Lemma~\ref{lem:blue-edge-uncrossed},
the total length of the paths needed to reconstruct these blue canonical edges can be
bounded by $3 \dst(v,u)$.  Also, since either the white canonical
edges themselves or their anchors are in $\spanner$, the total length
of the white canonical edges can be bounded by $2\dst(v,u)$ by
Lemma~\ref{lem:white-white}.  Therefore, $\dspan(v,u)$ can be bounded
by $5 \dst(v,u)$ for the edge $(v,u)$ as stated.  Furthermore,
if $(v,w) \in \spanner$, $\dspan(u,w)$ can be bounded by $5 \dst(v,u)
+ \dst(v,w)$, which in turn is bounded by $6 \dst(u,w)$.
\end{proof}

\begin{definition}
\label{def:white-dist}
For any two points $p,q \in \Points$ such that $p$ lies in a white
cone of $q$, we define $\dstw(p,q) = \dstw(q,p) = \dst(p,q) -
\dstb(p,q)$.
\end{definition}

\begin{lemma}
For any two white edges $(v,w),(u,w) \in \dela$ that both lie in the
same negative cone of $w$, if $u$ lies in a positive white cone of $v$, then
$\dstb(u,w) = \dstb(v,w) + \dst(v,u)$ and
$\dst(u,w) = \dst(v,w) + \dstb(v,u)$.
\label{lem:metric-props}
\end{lemma}
\begin{proof}
Since $(v,w) \in \dela$, we observe that $u$ must lie in the positive
white cone of $v$ that does not contain $w$.  The rest follows
directly from the definition of $\dstb$.
\end{proof}

\begin{lemma}
For any white anchor $(v,w)$, $\dspan(v,w) \leq 9 \dst(v,w)$.
Furthermore, for any uncrossed white edge $(u,w)$ in the fan of
$(v,w)$, we have
$\dspan(u,w) \leq 9 \dst(u,w) + \dstb(u,w)$ if $(u,w)$ lies on the
white side of $(v,w)$,
and $\dspan(u,w) \leq 9 \dst(u,w)$ otherwise.
\label{lem:white-edge-uncrossed}
\end{lemma}
\begin{proof}
If $(v,w) \in \spanner$, then clearly $\dspan(v,w) \leq \dst(v,w)$.
As for any uncrossed edge $(u,w)$ in the fan, by
Lemma~\ref{lem:white-white}, we get a bound of $2 \dst(u,w)$ on
$\dspan(u,w)$ if $(u,w)$ lies on the white side of $(v,w)$, and by
Lemma~\ref{lem:white-blue}, we get a bound of $6 \dst(u,w)$ on
$\dspan(u,w)$ if $(u,w)$ lies on the blue side of $(v,w)$.  Then, we
consider $(v,w) \notin \spanner$ and analyze two cases: $(v,w)$ was
not added in $\anchiv$ because of an adjacent anchor at $w$, or
because of an adjacent anchor at $v$.

If $(v,w)$ was not added in $\anchiv$ because of an adjacent (white)
anchor at $w$, let $(w,w')$ be that anchor (see Figure~\ref{fi:anchors}-(a)).
By our construction of
$\anchiv$, we know that $(w,w')$ must be shorter than $(v,w)$, \ie,
$\dst(w,w') < \dst(v,w)$.  Therefore, $v$ lies in the positive blue
cone of $w'$, and hence, there must be an outgoing blue edge at $w'$.
Let $(w',u')$ be that blue edge; then, $(u',w)$ must be a white
boundary edge of $(v,w)$, and possibly $u'=v$.
Using the fact that $u'$ lies in the positive blue cone of $w'$, it is
easy to verify that $\dst(v,u') \leq \dstw(v,w) \leq \dst(v,w)$.
Similarly, using the fact that $u'$ lies in a positive white cone of
$v$, it is easy to verify that $\dst(w',u') \leq \dst(v,w) +
\dstw(v,w) \leq 2\dst(v,w)$.
Following the path from $w$ to $w'$ to $u'$ to $v$, we bound
$\dspan(v,w)$ by $\dst(w,w')$ for the edge $(w,w')$, by $3\dst(w',u')$
for the edge $(w',u')$ using Lemmas~\ref{lem:uncrossed} and
\ref{lem:blue-edge-uncrossed}, and by $2\dst(v,u')$ for the edge
$(v,u')$ using Lemmas~\ref{lem:uncrossed} and \ref{lem:white-white}.
Also, using the above inequalities $\dst(w,w') < \dst(v,w)$,
$\dst(v,u') \leq \dst(v,w)$, and $\dst(w',u') \leq 2\dst(v,w)$, we get
the desired upper bound $\dspan(v,w) \leq 9 \dst(v,w)$.

Next, we consider the uncrossed white edges.  For any uncrossed edge
$(u,w)$ on the white side of the anchor, the bound on $\dspan(v,w)$
applies directly to $\dspan(u,w)$ because the path between $v$ and $w$
already connects $u$ and $w$.  Also, since $\dst(v,w) \leq \dst(u,w)$,
we immediately get $\dspan(u,w) \leq 9 \dst(u,w)$.
As for any white edge $(u,w)$ on the blue side of the anchor, we start
by observing that $\dstw(v,w) \leq \dst(u,w) - \dst(v,u)$, and that
$\dst(v,w) \leq \dst(u,w)$.  Also, using the above inequalities
$\dst(w,w') < \dst(v,w)$, and $\dst(v,u') \leq \dstw(v,w)$, and
$\dst(w',u') \leq \dst(v,w) + \dstw(v,w)$, we obtain the inequalities
$\dst(w,w') \leq \dst(u,w)$, and $\dst(v,u') \leq \dst(u,w) -
\dst(v,u)$, and $\dst(w',u') \leq 2 \dst(u,w) - \dst(v,u)$.  Then,
following the above path from $w$ to $v$ and extending it to $u$
using the canonical edges, we get $\dspan(u,w) \leq \dst(w,w') + 3
(\dst(w',u') + \dst(v,u)) + 2 (\dst(v,u') + \dst(v,u))$, which is then
bounded by $\dspan(u,w) \leq \dst(u,w) + 6 \dst(u,w) + 2 \dst(u,w) = 9
\dst(u,w)$.

\input{figures/proofs2-fig.tex}

In the case when $(v,w)$ was not added in $\anchiv$ because of an
adjacent (white) anchor at $v$,
(see Figure~\ref{fi:anchors}-(b))
let $(v',v)$ be that anchor.
%
%
Similarly, we know that $(v',v)$
must be shorter than $(v,w)$, \ie,
$\dst(v',v) < \dst(v,w)$.  Therefore, we know that there
must be an outgoing blue edge at $w$.  Let $(w,u')$ be that blue edge;
then, $(u',v)$ must be a white boundary edge of $(v',v)$, possibly
$u'=v'$.
Once again, it is easy to see that the length of the edge
$(v',u')$ cannot exceed the length on the white anchor $(v,w)$ and
that the length of the blue edge $(w,u')$ cannot exceed twice the
length of the white anchor $(v,w)$, \ie, $\dst(v',u') \leq \dst(v,w)$
and $\dst(w,u') \leq 2 \dst(v,w)$.
Following the path from $w$ to $u'$ to $v'$ to $v$, we bound
$\dspan(v,w)$ by $3\dst(w,u') + 2\dst(v',u') + \dst(v',v)$ using
Lemmas~\ref{lem:uncrossed}, \ref{lem:blue-edge-uncrossed} and
\ref{lem:white-white}.	Using the above inequalities, we get the
desired upper bound $\dspan(v,w) \leq 9 \dst(v,w)$.
As for the uncrossed white edges in this case, we first note that the
anchor $(v,w)$ itself is the boundary edge on the blue side and that
there are no white edges on the blue side, thus nothing to prove on
the blue side.
As for any uncrossed white edge $(u,w)$ on the white side of $(v,w)$,
we follow the path from $w$ to $v$ as above and to $u$ using the path
described in Lemma~\ref{lem:white-white}, which also bounds the last
step by $\dst(v,u)+\dstb(v,u)$.  We get an upper bound $\dspan(u,w)
\leq 9 \dst(v,w) + \dst(v,u) + \dstb(v,u)$.
Since $u$ lies in the white positive cone of $v$ that does not contain
$w$, by Lemma~\ref{lem:metric-props}, we know that $\dst(u,w) =
\dst(v,w) + \dstb(v,u)$, and that $\dst(v,u) \leq \dstb(u,w)$.
Therefore, we get $\dspan(u,w) \leq 9 \dst(u,w) + \dstb(u,w)$ as
desired.
\end{proof}

\begin{lemma}
For any crossed blue edge $(u,w) \in \dela$, $\dspan(u,w) \leq 3
\dst(u,w) + 9 \dstm(u,w)$.
\label{lem:blue-edge-crossed}
\end{lemma}
\begin{proof}
Let $(p,q)$ be a shortcut $(p,q)$ crosses $(u,w)$.  Since this
shortcut is in the blue cone, it must have been added in $\spanner$
replacing two white canonical edges incoming at $u$, namely $(p,u)$
and $(q,u)$.  As $p$ and $q$ are endpoints of the shortcut $(p,q)$,
both blue edges $(p,w)$ and $(q,w)$ are uncrossed.  Consequently, we
have $\dspan(p,w) \leq 3 \dst(p,w)$ and $\dspan(q,w) \leq 3 \dst(q,w)$
by Lemma~\ref{lem:blue-edge-uncrossed}.
Furthermore, by Lemmas~\ref{lem:uncrossed} and
\ref{lem:white-edge-uncrossed} we know that both of the canonical
edges $(p,u)$ and $(q,u)$ satisfy the inequalities $\dspan(p,u) \leq 9
\dst(p,u)$ and $\dspan(q,u) \leq 9 \dst(q,u)$.	Since both of these
canonical edges are incoming at $u$, we also know that one of them,
say $(p,u)$, is not longer than $\dstm(u,w)$, \ie, $\dst(p,u) \leq
\dstm(u,w)$.  Following the path via the point $p$, we bound the
distance $\dspan(u,w)$ by $3 \dst(p,w) + 9 \dst(p,u)$, which in turn
is bounded by $3 \dst(u,w) + 9 \dstm(u,w)$ as stated in the lemma.
\end{proof}

\begin{lemma}
\label{lem:white-edge-crossed}
For any crossed white edge $(u,w) \in \dela, \dspan(u,w) \leq 10 \dst(u,w)
+ 10 \dstm(u,w)$.
\end{lemma}
\begin{proof}
Since there are no shortcuts on the blue side of a
white anchor, all whites edges on the blue side are uncrossed.
Therefore, we can assume that $(u,w)$ is a white edge on the white
side crossing a shortcut $(p,q)$ in the same cone.

We consider two specific paths between $u$ and $w$:
{\em path~A} which starts with the path for $(p,w)$ as discussed in
Lemma~\ref{lem:white-edge-uncrossed} and then follows the canonical
path from $p$ to $u$; and
{\em path~B} which starts with the same path for $(p,w)$, and then
takes the shortcut from $p$ to $q$, and then follows the canonical
path from $q$ to $u$ in the other direction.
%
%
Considering both of the paths A and B, we bound $\dspan(p,w)$ by
$9 \dst(p,w) + \dstb(p,w)$ using Lemma~\ref{lem:white-edge-uncrossed}.
Also, for any blue canonical edge $(x,y)$ on the canonical path
between $p$ and $q$, we know that $\dspan(x,y) \leq 3\dst(x,y)$ by
Lemmas~\ref{lem:uncrossed} and \ref{lem:blue-edge-uncrossed}.
And, except for the first white canonical edge $(p,r)$, for any other
white canonical edge $(x',y')$ on this path, we know that
$\dspan(x',y') \leq 2\dst(x',y')$ by Lemmas~\ref{lem:uncrossed} and
\ref{lem:white-white}, because either these white canonical edges
themselves or their anchors are in $\spanner$.
The first white canonical edge $(p,r)$ and its anchor, however, may
both be excluded from $\spanner$ because of the shortcut $(p,q)$.
Thus, we can only bound $\dspan(p,r)$ by the worse bound of $9
\dst(p,r) + \dstb(p,r)$ using Lemmas~\ref{lem:uncrossed} and
\ref{lem:white-edge-uncrossed}.  Using the fact that this cost has a
much worse multiplicative constant than all other white canonical
edges, it can be verified that the worst case happens when $(p,r)$ is
the only white canonical edge on the path from $p$ to $u$.

We bound the total length of path A
by $9 \dst(p,w) + \dstb(p,w)$ for the edge $(p,w)$,
by $9 \dst(p,r) + \dstb(p,r) \leq 10 \dst(p,r)$ for the white
canonical edges between $p$ and $u$, and by $3 \dst(r,u)$ for the blue
canonical edges between $p$ and $u$.
Also, by using arguments analogous to those used in the proof of
Lemma~\ref{lem:metric-props}, we observe that $\dst(p,r) \leq \dst(p,u)$
and $\dst(r,u) \leq \dstw(p,u)$.  Therefore, the total cost
of path A is bounded by
\[\dspan(u,w) \leq 9 \dst(p,w) + \dstb(p,w) + 13 \dst(p,u) - 3 \dstb(p,u).\]
Since $(p,u)$ is a white canonical edge in a white cone, by
Lemma~\ref{lem:metric-props}, we can substitute
$\dst(p,u) = \dstb(u,w) - \dstb(p,w)$ and
$\dstb(p,u) = \dst(u,w) - \dst(p,w)$.  Hence, we get

\begin{align}
\dspan(u,w) &\leq 12 \dst(p,w) - 12 \dstb(p,w) + 13 \dstb(u,w)
 - 3\dst(u,w)  \nonumber \\
&= 12 \dstw(p,w) + 10\dstb(u,w) - 3\dstw(u,w) \nonumber \\
&= 12 \dstw(p,w) + 10\dst(u,w) - 13\dstw(u,w) \nonumber \\
&= 10\dst(u,w) + 10\dstw(u,w) + 12\dstw(p,w) -23\dstw(u,w).
\label{eqn:path-p}
\end{align}

By using the fact that $\dstw(p,w) \leq \dst(u,w)$ and rearranging
(\ref{eqn:path-p}) we get:

\begin{align*}
\dspan(u,w) &\leq 10 \dstb(u,w) + 20 \dstw(u,w) + 10\dst(u,w) +
2\dstw(p,w) - 23\dstw(u,w) \\
&\leq 10 \dst(u,w) + 10 \dstb(u,w) + 2 \dstw(p,w) - 3\dstw(u,w).
\end{align*}

Hence, we conclude that $\dspan(u,w) \leq 10 \dst(u,w) + 10
\dstm(u,w)$ if $2 \dstw(p,w) \leq 3 \dstw(u,w)$.
Therefore, we assume $3\dstw(u,w) < 2\dstw(p,w) \leq 2\dst(u,w) =
2\dstw(u,w) + 2\dstb(u,w)$, and derive $\dstw(u,w) < 2\dstb(u,w)$.

Next, we bound the total length of path B
by $9 \dst(p,w) + \dstb(p,w)$ for the edge $(p,w)$,
by $\dst(p,q)$ for the shortcut, by $2 \dst(u,q)$ for the white
canonical edges between $p$ and $u$, and by $3 \dst(u,q)$ for the blue
canonical edges between $p$ and $u$; thus in total by
\[
\dspan(u,w) \leq
9 \dst(p,w) + \dstb(p,w) + \dst(p,q) + 5 \dst(u,q).
\]

By the triangle inequality, we have
$\dst(p,q) \leq \dst(p,u) + \dst(u,q)$ and using once again that
$\dst(p,u) = \dstb(u,w) - \dstb(p,w)$,
we bound the total cost of path B by

\begin{equation}
\dspan(u,w) \leq 9 \dst(p,w) + \dstb(u,w) + 6 \dst(u,q).
\label{eqn:path-q}
\end{equation}

We further bound the cost of path B if $\dst(u,w) \geq \dst(q,w)$, in
which case we observe $\dst(u,q) \leq \dstw(u,w)$.  Using the
inequality $\dstw(u,w) < 2\dstb(u,w)$, we rearrange (\ref{eqn:path-q})

\begin{align*}
\dspan(u,w) &\leq 9 \dst(p,w) + \dstb(u,w) + 6 \dst(u,q) \\
	    &\leq 9 \dst(p,w) + \dstb(u,w) + 6 \dstw(u,w) \\
	    &\leq 10 \dst(u,w) + 5 \dstw(u,w) \\
	    &\leq 10 \dst(u,w) + 10 \dstb(u,w) \\
	    &\leq 10 \dst(u,w) + 10 \dstm(u,w).
\end{align*}

Therefore, we consider the remaining case $\dst(u,w) < \dst(q,w)$.  In
this case, by Lemma~\ref{lem:metric-props}, we get $\dstb(u,q) =
\dst(q,w) - \dst(u,w)$, and we also observe that $\dstw(u,q) \leq
\dstw(u,w)$.  Combining them, we get $\dst(u,q) \leq \dst(q,w) -
\dstb(u,w)$, and rearrange (\ref{eqn:path-q}) we get

\begin{equation}
\dspan(u,w) \leq 9 \dst(p,w) + 6\dst(q,w) - 5\dstb(u,w).
\label{eqn:path-q-worst}
\end{equation}

Furthermore, using the earlier assumption that $3\dstw(u,w) <
2\dstw(p,w)$, and the fact that $\dstw(q,w) \geq 0$, we deduce that
$3\dstb(u,q) < 2\dstb(p,u) + \dstb(u,q)$.  Using
Lemma~\ref{lem:metric-props} on all the three terms, we obtain
$\dst(q,w) < 3\dst(u,w) - 2\dst(p,w)$.
We then rearrange (\ref{eqn:path-q-worst})

\begin{align*}
\dspan(u,w) &\leq 9 \dst(p,w) + 18\dst(u,w) - 12\dst(p,w) - 5\dstb(u,w) \\
	    &= 10\dst(u,w) + 8\dstw(u,w) + 3\dstb(u,w) - 3\dst(p,w) \\
	    &\leq 10\dst(u,w) + 8\dstw(u,w) + 3\dstb(u,w)-3\dstw(p,w)\\
	    &\leq 10\dst(u,w) + \frac{7}{2}\dstw(u,w) + 3\dstb(u,w)\\
	    &\leq 10\dst(u,w) + 10\dstb(u,w).
\end{align*}

Therefore, what remains to be proven is $\dspan(u,w) \leq 10\dst(u,w)
+ 10\dstm(u,w)$ for the case $\dstm(u,w) = \dstw(u,w)$.  To prove this
case, using (\ref{eqn:path-p}), we will assume the stronger inequality
$7\dstw(u,w) < 4\dstw(p,w)$, and using arguments similar to those used
in the above paragraph, we get $3\dst(q,w) < 7\dst(u,w) - 4\dst(p,w)$.
Once again, we rearrange (\ref{eqn:path-q-worst})

\begin{align*}
\dspan(u,w) &\leq 9 \dst(p,w) + 14\dst(u,w) - 8\dst(p,w) - 5\dstb(u,w) \\
	    &= 9\dst(u,w) + \dst(p,w) + 5\dstw(u,w) \\
	    &\leq 10 \dst(u,w) + 5 \dstw(u,w).
\end{align*}

Proving $\dspan(u,w) \leq 10 \dst(u,w) + 10 \dstm(u,w)$ for all of the
cases, we conclude that the statement of the lemma is correct.

\end{proof}

 By Lemmas~\ref{lem:blue-edge-uncrossed},
\ref{lem:white-edge-uncrossed}, \ref{lem:blue-edge-crossed}, and
\ref{lem:white-edge-crossed} we have that
 $\spanner$ is a $20$-spanner of $\dela$. Since $\dela$ is a $2$-spanner
 of $\comp$ (\cite{Che89}) it follows that $\spanner$ is a spanner of
 $\comp$ with stretch factor at most $40$. We prove, however, a much better bound of 20 next.

\begin{lemma}
For any two points $p,q \in \Points$, $\dspan(p,q)$ is bounded by $20
\dst(p,q)$.
\label{lem:stretch}
\end{lemma}
\begin{proof}
We prove the lemma by first constructing, in $\dela$, a monotone path
$\pi$ between $p$ and $q$ that lies inside $\tridown(p,q)$. We then
consider the path $\pi'$ in $\spanner$ obtained by replacing every
edge of $\pi$ not in $\spanner$ with a short path in $\spanner$.
%

We define the path $\pi$ between $p$ and $q$ consisting of $k$
edges in $\dela$ using a sequence of pairs of points $\{p, q\} = \{p_0,
q_0\}, \{p_1, q_1\}, \ldots, \{p_k, q_k\}$ such that any two
consecutive pairs of points $\{p_{i-1}, q_{i-1}\}$ and $\{p_i, q_i\}$
satisfy exactly one of the equations $p_i = p_{i-1}$ and $q_i =
q_{i-1}$ and the equation that is not satisfied describes the $i^{th}$
edge.  If $p_i \neq p_{i-1}$, then the $i^{th}$ edge is $(p_{i-1},
p_i) \in \dela$, otherwise the $i^{th}$ edge is $(q_{i-1}, q_i) \in
\dela$.
We define this sequence recursively for the next pair of points
$\{p_{i+1}, q_{i+1}\}$ by first identifying which of the points $p_i$
and $q_i$ lie in the other's positive cone.  If $q_i$ lies in the
positive cone of $p_i$, then we define $q_{i+1} = q_i$ and $p_{i+1}=r$
such that $(p_i,r) \in \dela$, noting that by definition of $\dela$,
$r$ is the unique such point in the positive cone of $p_i$ that
contains $q_i$.  Otherwise, if $p_i$ lies in a positive cone of $q_i$,
then we define $p_{i+1} = p_i$ and $q_{i+1}=r'$ such that $(q_i ,r')
\in \dela$.
We stop when $p_k =q_k$.

We prove inductively that the aforementioned path $\pi$ lies within
$\tridown(p,q)$.
  For the base case, clearly the path consisting of
 the only edge in the sequence $\{p_{k-1},q_{k-1}\}, \{p_k,q_k\}$ lies
 within $\tridown(p_{k-1},q_{k-1})$.  For the inductive step, assuming
 that the path for the sequence $\{p_i,q_i\}, \ldots, \{p_k, q_k\}$
 lies within $\tridown(p_i,q_i)$, we show that the path for the
 sequence $\{p_{i-1},q_{i-1}\}, \{p_i,q_i\}, \ldots, \{p_k,q_k\}$ lies
 within $\tridown(p_{i-1},q_{i-1})$.  First, we observe that in either
 case that the first edge is $(p_{i-1},p_i)$ or $(q_{i-1},q_i)$, it
 lies within $\tridown(p_{i-1},q_{i-1})$ by definition.	Finally, we
 observe that $\tridown(p_i,q_i)$, hence the rest of the path lies
 within $\tridown(p_{i-1},q_{i-1})$ as well.  Therefore, we prove the
 inductive step.

We then prove that all of the edges of the form $(p_i, p_{i+1})$ lie
in the same corresponding positive cones of their respective points
$p_i$.	More specifically, we prove by induction on the sequence of
such edges $e_0=(p=p_{i_0-1},p_{i_0}), e_1=(p_{i_0}=p_{i_1-1},
p_{i_1}), \ldots, e_\ell=(p_{i_{\ell-1}}=p_{i_\ell-1},
p_{i_\ell}=p_k)$, that $p_{i_0}, p_{i_1}, \ldots, p_{i_\ell}$ lie in
the same corresponding cones of $p_{i_0-1}, p_{i_1-1}, \ldots,
p_{i_\ell-1}$ respectively.

The base case follows trivially, and for the inductive step we assume
for edges $e_0, e_1, \ldots, e_\lambda$ that $p_{i_0}, p_{i_1},
\ldots, p_{i_\lambda}$ lie in the same corresponding cones of
$p_{i_0-1}, p_{i_1-1}, \ldots, p_{i_\lambda-1}$ respectively.  For the
inductive step we need to prove for the edge $e_{\lambda+1}$ that
$p_{i_{\lambda+1}}$ lies in the same corresponding cone of
$p_{i_{\lambda+1}-1}=p_{i_{\lambda}}$.
 We have already proven that the edge $e_{\lambda+1}$ lies within
$\tridown(p_{i_{\lambda}}, q_{i_{\lambda}})$.  Also, by definition of
$\dela$, we know that $\tridown(p_{i_\lambda-1}, p_{i_{\lambda}})$ is
empty of points of $\Points$ in its interior.
Then we conclude the inductive proof by observing that the only
positive cone that can possibly include $e_{\lambda+1}$ at
$p_{i_{\lambda}}$ is part of the same corresponding cone of
$p_{i_{\lambda}}$.
Having proven this critical property about this path, we denote it by
$\pi_p$ and refer to it as one of the two {\em branches}, where the
other {\em branch} $\pi_q$ is defined analogously using points
$q_0, q_1, \ldots, q_k$.
We conclude that the path $\pi$ consisting of these two branches
$\pi_p$ and $\pi_q$ is monotone.

Finally, we prove the claimed bound on the length of the path $\pi'$
between $p$ and $q$. Because the constants in Lemma~\ref{lem:white-edge-crossed} are the largest among Lemmas~\ref{lem:blue-edge-uncrossed},
\ref{lem:white-edge-uncrossed}, \ref{lem:blue-edge-crossed}, and
\ref{lem:white-edge-crossed}, the worst case happens when $q$ lies in the white positive cone of $p$ and $\pi$ is white monotone.
Let $z$ be the blue vertex of $\tridown(p,q)$.
By Lemma~\ref{lem:monotone}: the projections of all edges of $\pi$
onto $zp$ (resp.~$zq$) do not overlap, are contained within $[zp]$ (resp.~$zq$), $|zp|=\dst(p,q)$, and $|zq| =\dstb(p,q)$. In the worst case, each edge of $\pi$ is crossed, and Lemma~\ref{lem:white-edge-crossed} applies
to reconstruct each edge of $\pi$. Therefore, the length of $\pi'$ can be upper bounded by $10\dst(p,q) + 10\dstm(p,q) \leq 10\dst(p,q) + 10\dstb(p,q) \leq 20 \dst(p,q)$.
\end{proof}

\begin{theorem}
$\spanner$ is a plane spanner of $\comp$ with maximum degree 4 and
stretch factor at most 20 that can be constructed in $\Oh(n\log n)$ time.
\label{thm:main}
\end{theorem}
\begin{proof}
The properties of $\spanner$ regarding planarity, maximum degree
bound, and stretch factor are proven by Lemmas \ref{lem:plane},
\ref{lem:degree}, and \ref{lem:stretch} respectively.  The TD-Delaunay
triangulation $\dela$ can be constructed in $\Oh(n \log n)$
time~\cite{Che89}.
Given $\dela$ and the fact that it is planar, $\spanner$ can be constructed
in $\Oh(n \log n)$ time: sorting the anchors takes $\Oh(n \log n)$ time, and
adding edges to $\spanner$ can be done in $\Oh(n)$ time.
\end{proof}

%% file: figures/proofs2-fig.tex
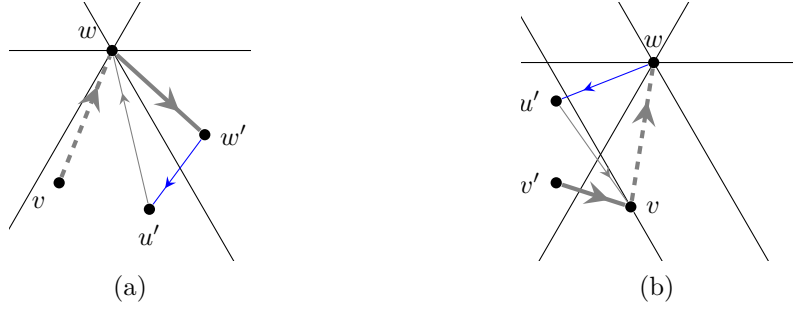
\begin{figure}
\noindent
\begin{minipage}{.5\textwidth}
\centering
\begin{tikzpicture}[scale = 0.8]

\clip(-1.7,-3.5) rectangle (2.3,0.8);

\draw (-5,0) -- (5,0);
\draw (0,0) -- (60:5); \draw (0,0) -- (240:5);
\draw (0,0) -- (120:5); \draw (0,0) -- (300:5);

\draw (0,0) node[vertex,label=150:{$w$}] (w) {};
\draw (1.54,-1.4) node[vertex,label=0:{$w'$}] (wp) {};
\draw (0.62,-2.64) node[vertex,label=-90:{$u'$}] (up) {};
\draw (-0.88,-2.2) node[vertex,label=-135:{$v$}] (v) {};

\draw [canonical,blue] (wp) -- (up);
\draw [anchorEdge] (w) -- (wp);
\draw [anchorEdge, dashed] (v) -- (w);
\draw [canonical] (up) -- (w);

\end{tikzpicture}

\centering (a)
\end{minipage}%
\begin{minipage}{.5\textwidth}
\centering
\begin{tikzpicture}[scale = 0.8]

\clip(-2.2,-3.3) rectangle (2.3,1);

\draw (-5,0) -- (5,0);
\draw (0,0) -- (60:5); \draw (0,0) -- (240:5);
\draw (0,0) -- (120:5); \draw (0,0) -- (300:5);

\draw (0,0) node[vertex,label=90:{$w$}] (w) {};
\draw (-1.62,-2) node[vertex,label=180:{$v'$}] (vp) {};
\draw (-1.62,-0.64) node[vertex,label=180:{$u'$}] (up) {};
\draw (-0.38,-2.4) node[vertex,label=0:{$v$}] (v) {};

\draw (v) -- ++(120:5); \draw (v) -- ++(300:5);


\draw [canonical,blue] (w) -- (up);
\draw [canonical] (up) -- (v);
\draw [anchorEdge] (vp) -- (v);
\draw [anchorEdge, dashed] (v) -- (w);


\end{tikzpicture}

\centering (b)
\end{minipage}

\caption{Illustrations of the proof of Lemma~\ref{lem:white-edge-uncrossed}.
(a) The case when $(v,w) \notin \anchiv$ because a shorter adjacent
anchor $(w,w')$ was added first. Edge $(w',u')$ is a blue boundary edge and
there is a
white monotone path from $u'$ to $v$ in $\spanner$. (b) The case when
$(v,w) \notin
\anchiv$ because a shorter adjacent anchor $(v',v)$ was added
first. Edge $(w,u')$ is a blue boundary edge and $(u',v)$ is a white boundary
edge on the white side of its cone and there is a white monotone path between $u'$
and $v'$ in $\spanner$.}
\label{fi:anchors}
\end{figure}

%% file: convex.tex
\section{Tight degree bound for points in convex position}
\label{sec:convex}
In this section, we show that if the set $\Points$ of points is in convex position, then the {\em same} spanner $\spanner$ constructed in the previous section has maximum degree at most 3. Therefore, for any set of points $\Points$ in convex position,
there exists a plane spanner of $\comp$ of maximum degree at most 3.  We also show in this section that 3 is a lower bound on the maximum degree of plane spanners of $\comp$ for point-sets in convex position. The preceding implies that 3 is a tight bound on the maximum degree of plane spanners of $\comp$ for point-sets in convex position, and this completely and satisfactorily answers the question about the maximum degree of plane geometric spanners of $\comp$, in the case when $\Points$ is in convex position.

\begin{proposition}
\label{prop:convexupper}
Let $\Points$ be a set of points in convex position in the plane, and let $\spanner$ be the spanner of $\comp$ constructed as described in Section~\ref{sec:spanner}. Then the maximum degree of $\spanner$ is at most 3.
\end{proposition}

\begin{proof}
Let $p \in \Points$. It suffices to show that the degree of $p$ in $\spanner$ is at most 3. To this effect, we show that at most 3 edges incident to $p$ could be added in steps 1 -- 6 of the construction of $\spanner$.   Since $\Points$ is in convex position, there exists a support line, $D_p$, passing through $p$ such that all the points of $\Points$ lie in
one closed half plane, $H$, of the two half planes delimited by $D_p$~\cite{preparatashamos}.  Observe that -- by the construction of $\spanner$ -- each of the two blue sectors of $p$ contains at most one edge incident to $p$ in $\spanner$. In Lemma~\ref{lem:charging}, we showed that each of the at most 4 edges incident to a point $p \in \spanner$ is charged by the charging scheme (see Section~\ref{sec:spanner-props}) to one of the 4 sectors $LB_p$, $RB_p$, $UW_p$, and $LW_p$, such that each of the 4 sectors is charged with at most one edge. In this charging scheme, a blue anchor at $p$ is charged to the blue sector containing it, and a white edge incident to $p$ is either charged to a blue sector at $p$ (and in such case the blue sector does not contain a blue anchor), or to the white sector containing the edge. We distinguish the following cases, based on the angle, $\alpha$, that $D_p$ makes with the positive $x$-axis:

\begin{itemize}
\item[Case 1.]  $0 \leq \alpha \leq \pi/3$ ($D_p$ passes through the two blue sectors at $p$). In this case the half plane $H$ entirely contains one of the two white sectors at $p$. Suppose that $H$ contains $UW_p$; the case is symmetric if $H$ contains $LW_p$.
Since (1) each of the three sectors $LB_p$, $UW_p$, and $RB_p$ is charged with at most one edge incident to $p$, (2) each of $LB_p$ and $RB_p$ contains at most one blue anchor incident to $p$, and (3) each edge incident to $p$ in $UW_p$ is either charged to $UW_p$, $LB_p$, or $RB_p$, it follows that the number of edges incident to $p$ in $H$, and hence in $\spanner$, is at most 3.

\item[Case 2.] $\pi/3 \leq \alpha \leq 2\pi/3$ ($D_p$ passes through the two red cones at $p$). Assume, to get a contradiction, that there are 4 edges incident to $p$ in $H$. Suppose first that
               $H$ entirely contains $RB_p$, and hence is disjoint from $LB_p$. In this case the part of $UW_p$ in $H$ is contained in a positive cone at $p$, and hence, there is at most one edge incident to $p$ in $\spanner$ that lies in $H \cap UW_p$. Since each of the 4 sectors at $p$ is charged with at most one edge, it follows from the preceding that there is a white edge $e$ incident to $p$ in $H \cap LW_p$ that is charged to $LB_p$. This could only happen if $e$, a white edge, is a canonical edge in a blue cone, added according to step 3 in the spanner construction, and is charged to $LB_p$ according to step 3 of the charging scheme. This, however, implies that $LB_p$ is not empty (a white canonical edge $e \in LW_p$ in a blue cone could exist only if $RB_p$ contains points of $\Points$), contradicting that all points of $\Points$ lie in $H$. Suppose now that $H$ entirely contains $LB_p$, and hence is disjoint from $RB_p$. By the same token as above, there must exist a white canonical edge, $e$, incident to $p$ and lying in a white sector at $p$, that is charged to $LB_p$. By the charging scheme, the edge $e$ cannot lie in $UW_p$ because in such case $e$ would not be charged to $UW_p$. Therefore, $e$ must lie in $LW_p$. But then $e$ must lie in the negative cone adjacent to $p$ in $LW_p$, which is (the cone) disjoint from $H$, again contradicting that $H$ contains all points of $\Points$.

\item[Case 3.] $2\pi/3 \leq \alpha \leq \pi$ ($D_p$ passes through the two green cones at $p$). The proof is analogous to that of Case 2 above.

\end{itemize}
\end{proof}

\begin{proposition}
\label{lem:convexlower}
For any constant $\rho \geq 1$, there exists a point-set $\Points$ in convex position such that any plane spanner of $\comp$ of maximum degree at most 2 has stretch factor $> \rho$.
\end{proposition}

\begin{proof}
Let $\rho \geq 1$ be a given constant. Choose an integer $b > \rho$, and an integer $N > 3(\rho \cdot b +1)$. Consider an orthogonal rectangle of vertical dimension $a=N-1$ and horizontal dimension $b$. Let $n=2N$, and let $\Points =\{p_1, \ldots, p_N, q_1, \ldots, q_N\}$ be a set of $n$ points placed on the rectangle as follows. Points $p_1, \ldots, p_N$ are placed on one vertical side of the rectangle such that $|p_ip_{i+1}| =1$, for $i=1, \ldots, N-1$, so that $p_1$ and $p_N$ end up on the two vertices of the vertical side of the rectangle. Points $q_1, \ldots, q_N$ are placed symmetrically on the other vertical side of the rectangle so that the $p_iq_i$'s, $i=1, \ldots, N$, are all parallel. We note that one can choose $\Points$ to be in convex position while
respecting the standard general position assumptions (no 3 points on a line, etc.) by slightly modifying the set of points chosen above (e.g., rotating each of $p_ip_{i+1}$ and $q_iq_{i+1}$, $i=1, \ldots, N$, slightly but increasingly inwards (\ie, towards the interior of the rectangle), and slightly modifying the choice of the parameters $a, b, N$. However, we decided to go with the above configuration for clarity and ease of presentation.

Suppose, to get a contradiction, that there is a plane spanner $S$ of $\comp$ of maximum degree at most 2 and stretch factor $\rho$. Since $S$ is connected, $S$ is either a Hamiltonian path or a Hamiltonian cycle. Without loss of generality, we will assume in what follows that $S$ is a Hamiltonian cycle, as the proof in the other case is simpler. Let $L=\{p_1, \ldots, p_N\}$ and $R=\{q_1, \ldots, q_N\}$. We distinguish the following two cases:

\begin{itemize}
\item[Case 1.] The subgraph of $S$ induced by either $L$ or $R$ is disconnected. Suppose that the subgraph of $S$ induced by $L$, $S_{L}$, is disconnected; the proof follows by symmetry in the other case.
Then there must exist $i \in \{1, \ldots, N-1\}$ such that the two consecutive points $p_i, p_{i+1}$ belong to two different connected components of $S_L$. Therefore, any path between $p_i$ and $p_{i+1}$ in $S$ must contain an edge $p_rq_s$, $r, s \in \{1, \ldots, N\}$. Since the distance between any point from $L$ and any point from $R$ is at least $b$, we have $|p_rq_s| \geq b$. It follows that the length of any path in $S$ between $p_i$ and $p_{i+1}$ is at least $b > \rho =\rho \cdot |p_ip_{i+1}|$, which contradicts that $S$ is a spanner of stretch factor $\rho$.

\item[Case 2.] Each of the two subgraphs $S_L$ and $S_R$ of $S$, induced by $L$ and $R$, respectively, is connected. Then each of $S_L$ and $S_R$ must be a Hamiltonian path on $L$ and $R$, respectively.
Let $p_r$ and $p_s$, $r, s \in \{1, \ldots, N\}$, be the points of degree 1 in $S_L$, and $q_{r'}, q_{s'}$, $r', s' \in \{1, \ldots, N\}$, be those of degree 1 in $S_R$. Since $S$ is a Hamiltonian cycle, $S$ consists of the edges in the Hamiltonian path $S_L$, the edges in the Hamiltonian path $S_R$, plus a matching between $\{p_r, p_s\}$ and $q_r', q_s'$, say $p_rp_{r'}$ and $q_sq_{s'}$. Since there are $N$ points in $L$, there must exist a point $p_i \in L$, $i \in \{1, \ldots, N\}$, such that the number of edges on each of the two subpaths between $p_i$ and $p_r$ and between $p_i$ and $p_s$ in $S_L$, is at least $N/3 -1 > \rho \cdot b$ by the choice of $N$. Consider point $q_i \in R$. Any path between $p_i$ and $q_i$ in $S$ must contain either the subpath of $S_L$ between $p_i$ and $p_r$, or the subpath of $S_L$ between $p_i$ and $p_s$, and hence must contain more than $\rho \cdot b$ edges of $S_L$. Since each edge of $S_L$ has length at least 1, the length of any path between $p_i$ and $q_i$ in $S$ is more than $\rho \cdot b =\rho \cdot |p_iq_i|$. This again contradicts the assumption that $S$ has stretch factor $\rho$, and completes the proof.
\end{itemize}
\end{proof}

We conclude with the following theorem:

\begin{theorem}
\label{thm:convex}
The constant 3 is a tight bound on the maximum degree of geometric plane spanners of $\comp$ for point-sets in convex position.
\end{theorem} 